\documentclass[journal]{IEEEtran}
\usepackage{cite}

\ifCLASSINFOpdf
  \usepackage[pdftex]{graphicx}
\else
  \usepackage[dvips]{graphicx}
\fi
\usepackage{amsmath}

\usepackage[utf8]{inputenc}

\usepackage{amssymb}
\usepackage{amsfonts}
\usepackage{amsthm}
\usepackage{color}
\usepackage{braket}
\usepackage[normalem]{ulem}
\usepackage{times}
\usepackage{subfigure}
\usepackage{pict2e}
\usepackage{bm}
\usepackage{slashbox}
\usepackage{tabularx}
\usepackage{multirow,bigdelim}
\usepackage{enumitem}
\usepackage[colorlinks,linkcolor=magenta,citecolor=blue]{hyperref}

\newcommand{\sign}{\mathrm{sgn}}

\newcommand{\trans}[1]{{#1}^{\ensuremath{\mathsf{T}}}} 
\newcommand{\Tr}{\mathrm{Tr}}

\newcommand{\betaout}{\beta_{\rm out}}
\newcommand{\betaoutmod}{\bar{\beta}_{\rm out}}

\newcommand{\CapG}{C^{\mathrm{G}}}
\newcommand{\CapGchi}{C^\mathrm{G}_\chi}

\newcommand{\chG}{\Phi}

\newcommand{\chPS}{\chG^{\rm F}}

\newcommand{\chiG}{\chi^{\mathrm{G}}}

\newcommand{\disp}{\bm{d}}
\newcommand{\dispin}{\disp_{\rm in}}

\newcommand{\dispout}{\disp_{\rm out}}
\newcommand{\dispch}{\disp_{\rm env}}
\newcommand{\CM}{\bm{V}}

\newcommand{\CMin}{\CM_{\rm in}}

\newcommand{\CMinmod}{\bar{\CM}_{\rm in}}
\newcommand{\CMenv}{\CM_{\rm env}}
\newcommand{\CMout}{\CM_{\rm out}}
\newcommand{\CMoutmod}{\bar{\CM}_{\rm out}}

\newcommand{\eff}{\bm{Y}}

\newcommand{\effPS}{\eff_{\rm F}}

\newcommand{\fr}{\omega}
\newcommand{\frenv}{\omega_{\rm env}}

\newcommand{\frthr}{\omega_{\rm thr}}
\newcommand{\frin}{\omega_{\rm in}}
\newcommand{\frininf}{\omega_{\infty}}
\newcommand{\frintauinf}{\omega_{\rm in}^{(\tau_\infty)}}
\newcommand{\frinmod}{\bar{\omega}_{\rm in}}
\newcommand{\frout}{\omega_{\rm out}}

\newcommand{\froutmod}{\bar{\omega}_{\rm out}}

\newcommand{\M}{M}
\newcommand{\Menv}{M_{\rm env}}
\newcommand{\Menvcrit}{M_{\rm c}}

\newcommand{\Minmod}{\bar{M}_{\rm in}}
\newcommand{\Mout}{M_{\rm out}}
\newcommand{\Moutmod}{\bar{M}_{\rm out}}

\newcommand{\epp}{e_{\rm p}}

\newcommand{\Nthr}{\bar{N}_{\rm thr}}

\newcommand{\Ninmod}{\bar{N}}
\newcommand{\Ninmodcrit}{\bar{N}_{\rm c}}
\newcommand{\Nenv}{N_{\rm env}}

\newcommand{\Noutmod}{\bar{N}_{\rm out}}

\newcommand{\se}{\nu}

\newcommand{\Sim}{\bm S}

\newcommand{\X}{\bm{X}}

\newcommand{\XPS}{\X_{\rm F}}

\newcounter{thm} 
\newcounter{corr} 
\newcounter{lem} 
\newtheorem{lemPipeline}[lem]{Lemma}

\begin{document}
%
\title{CLASSICAL CAPACITY OF PHASE-SENSITIVE GAUSSIAN QUANTUM CHANNELS}
%
%
%

\author{Joachim~Sch\"afer, 
        Evgueni~Karpov, 
        Oleg~V.~Pilyavets,
        and~Nicolas~J.~Cerf
\thanks{This work was supported by the
Belgian federal government via the IAP research network Photonics$@$be grant number P7/35,
by the Belgian foundation FRIA, and
by the ULB funding ``Bourses Post-doctorales IN -- Ouvertures internationales.''}
\thanks{J.~Sch\"afer and O.~V.~Pilyavets were with
the Universit\'e libre de Bruxelles, 50 avenue F.~D.~Roosevelt, 1050 Ixelles,
Brussels, Belgium (e-mail: mail@joachimschaefer.de; pilyavets@gmail.com).}
\thanks{E.~Karpov and N.~J.~Cerf are with
the Universit\'e libre de Bruxelles, 50 avenue F.~D.~Roosevelt, 1050 Ixelles,
Brussels, Belgium (e-mail: \mbox{ekarpov@ulb.ac.be}; ncerf@ulb.ac.be).}
}

%
%





\markboth{CLASSICAL CAPACITY OF PHASE-SENSITIVE GAUSSIAN QUANTUM CHANNELS}%
{SCH\"AFER \MakeLowercase{\textit{et al.}}:
CLASSICAL CAPACITY OF PHASE-SENSITIVE GAUSSIAN QUANTUM CHANNELS}


%



\maketitle

\begin{abstract}
The full solution of the optimization problem giving the
Gaussian capacity of the single-mode fiducial Gaussian quantum
channel is provided. Since it was shown that the Gaussian
capacity of an arbitrary (phase-sensitive or insensitive)
single-mode Gaussian quantum channel is equal to the
Gaussian capacity of this fiducial channel, the solution
presented in this work can be regarded as universal. The
analytical study of this solution, below and above the energy
threshold, shows that the dependence of the Gaussian capacity
on the environment noise squeezing is not monotonic. In
particular, the capacity may have a saddle point, one or two
extrema at finite squeezing, or be a monotonically increasing
function of the squeezing parameter. The exact dependence is
defined by the determinant of the noise covariance matrix and
by the transmissivity (or gain) of the fiducial Gaussian
channel. 
\end{abstract}


%
\IEEEpeerreviewmaketitle

\section{Introduction}

%
%
%
%

 

\IEEEPARstart{I}{nformation} transmission through noisy communication channels
is a key problem of Information Theory. When the quantum nature of information carriers is taken into account, their interaction with the noisy environment is
properly described in terms of quantum channels. The classical capacity, i.e.,
the maximal rate of errorless transmission of classical information via a
quantum channel is one of the important figures of merit. 
In general, finding this quantity
is 
a highly non-trivial optimization problem. Nevertheless, this
problem was recently solved for a wide class of quantum channels with a Gaussian
noise~(so-called Gaussian quantum channels)~\cite{GGCH13,GHG13} which 
are good models of common physical systems such
as free space or fiber optics communications 
(see, e.g.,~\cite{CD94}).

The Gaussian quantum channels attracted much attention in the last
decades~\cite{GGCH13,GHG13,CD94,YuenOzawa,
GGLMS04a,H98,HSH99,HW01,SH02,SH03,GGLMS04b,BDM05,GM05,
CCMR05,H06,CGH06,EW07,H07,SH07,WHTH07,GSE07,PZM08,LPM09,SDKC09,YS80,PLM09,
SKC10,GHLM10,G11,SKC11,ISS11,GNLSC12,SKC12,GLMS12,WPGCRSL12,KS13,KG13,H14,
BradlerAdamiPRD,BardhanIEEE,BardhanIEEEConf,BradlerJPA,AIP2014,SKPPC13,S13};
however, 
only a few results on their classical capacity were known~\cite{YuenOzawa,
GGLMS04b,LPM09}.
One of the attempts to approach this problem  was to restrict 
the search for optimal states and encodings by Gaussian states and Gaussian
distributions~\cite{EW07}.  
The solution of such reduced optimization problem 
gives the so-called Gaussian classical
capacity~\cite{SKPPC13,H14}
which was actively studied during the last years~\cite{LPM09,EW07,SKPPC13,
GM05,HSH99,SH02,SH03,CCMR05,H06,SH07,PZM08,SDKC09,PLM09,SKC10,SKC11,SKC12,
AIP2014,S13}. Although the
Gaussian capacity\footnote{Since other capacities are not
disscussed in this work, the word ``classical'' will be usually omitted.}
is \emph{a priori} a lower bound on the actual 
capacity this quantity is of great interest in practice,
because the Gaussian states are more accessible in experimental
realizations.

Recently, the progress
in finding the Gaussian capacity 
became even more important
because it was proved that the classical capacity of bosonic Gaussian quantum
channels with thermal noise is indeed achieved by Gaussian ensembles of (pure)
Gaussian input states~\cite{GGCH13,GHG13}.\footnote{The argument of that proof was
based on the 
minimum output entropy conjecture \cite{GGLMS04a,GSE07,GHLM10,GNLSC12} which was also
recently proved~\cite{GGCH13,GHG13}.} 
If input energy is sufficiently high,
this result also leads to finding the classical capacity
in the case of squeezed thermal noise~\cite{GGCH13,H14}.
In both cases\footnote{They are known as the third stage cases~\cite{PLM09}.}
the calculations were technically performed in terms of the Gaussian
capacity.  However, the equivalence of the Gaussian
and classical capacity is not yet known for small input energies.\footnote{The minimum
output entropy conjecture does not help to find the proof for this case.} 
The study of Gaussian capacity in this
region~\cite{LPM09,SKC10,SKC11,PLM09,S13,SKC12,AIP2014}
provides a way to find a lower bound on the classical capacity with
the hope that 
equivalence of this bound and the (actual) classical capacity
will be eventually proven as well. In addition,
the analysis of the properties of
the Gaussian capacity may give a further insight helping to find such a proof.

In this work we focus on the Gaussian classical capacity of single-mode Gaussian quantum 
channels and 
calculate the one-shot (or single-letter)
Gaussian capacity of the fiducial channel~\cite{SKPPC13}.
Recently, it was shown that this
is equivalent to finding the
Gaussian capacity of an arbitrary single-mode Gaussian channel~\cite{SKPPC13,S13}.
This equivalence
is based on the fact that any single-mode Gaussian channel can be expressed as
the fiducial channel preceded by a rotation and followed by a symplectic
transformation (in some cases---only if a proper limit is also taken). 
We also go beyond the known results on
the Gaussian capacity and present the solution which is valid for all
values of input energy.
Then, we analyze the obtained Gaussian capacity as a function of the
channel parameters. In particular, we find its rich behavior as a function of
the noise squeezing. The present work generalizes previous results obtained for
the single-mode lossy (or attenuation)~\cite{LPM09,PZM08,PLM09}
and additive noise~\cite{SDKC09,SKC10,SKC11} Gaussian channels.

The paper is structured as follows. In Section~\ref{sec:chgauss}, we briefly
review Gaussian quantum states and Gaussian channels.
In Section~\ref{sec:capgauss}, we introduce the single-mode fiducial Gaussian
channel and calculate its Gaussian capacity. In this calculation we use a \emph{frequency
representation}\footnote{It was inspired by $\lambda$-representation introduced
in~\cite{PLM09} for lossy channel.} for the covariance matrices
to get a better insight about the physical meaning of the solution. 
In Section ~\ref{sec:chpara}, the
behavior of the capacity as a function of the channel parameters is analyzed.
Section~\ref{sec:conclusions} summarizes the results.

\section{Gaussian channels and Gaussian states}\label{sec:chgauss}

Quantum mechanics in continuous variables can be introduced as Weyl calculus
on the system's phase space $(q,p)$~\cite{DeGosson} (below we consider only
single-mode case, i.e., 1-D harmonic oscillator).  In this representation a
quantum state is specified by its Wigner function $W(\mathbf{x})$
[$\mathbf{x}=(q,p)$ is a vector of canonical variables $q$ and~$p$], which is a
Weyl symbol of the density operator $\hat\rho$ acting on the Hilbert
space\footnote{We assume that canonical commutation relation for quadratures
operators $\hat q$ and $\hat p$ belonging to $\mathcal{H}$ is $[\hat{q},\hat{p}]
= i$ ($\hbar=1$), and normalization of a Wigner function is $\int
W(\mathbf{x})\,d\,\mathbf{x}=2\pi$.} $\mathcal{H}\equiv L_2(\mathbb{R})$.

A Gaussian probability density
density
\begin{equation*}
W(\mathbf{x})=\frac1{\sqrt{\det\mathbf{V}}}e^{-\frac12\left(
\mathbf{x}-\mathbf{d},\mathbf{V}^{-1}(\mathbf{x}-\mathbf{d})
\right)}
\end{equation*}
with the mean $\mathbf{d}$ 
is the Wigner function of a quantum state if $\mathbf{V}$ is 
a quantum covariance matrix (see, e.g.,~\cite{DeGosson2}), i.e.,
its symplectic eigenvalue $\se=\sqrt{\det{\CM}}$ satisfies $\se\ge1/2$~\cite{DeGosson}.
Such quantum states are called Gaussian. Since their higher momenta are
functions of the first two, these states are completely specified by the mean
value $\mathbf{d}$ and the covariance matrix $\mathbf{V}$.


The Gaussian channel $\chG$ is defined as a completely positive trace preserving map (CPTP) which maps Gaussian states to Gaussian states \cite{H07}. Its action on the first two momenta of any state is completely defined by the real matrix $\X$, a
displacement $\dispch$ and the real, symmetric and positive matrix $\eff$, where $\eff$ can be considered as the covariance matrix of an \emph{effective noise}. The three parameters determine the action of $\chG$ on an input state with displacement $\dispin$
and covariance matrix $\CMin$ leading to the output state determined by displacement $\dispout$ and covariance matrix $\CMout$, i.e. \cite{CGH06}
\begin{equation}\label{eq:GCH} 
	\begin{split} 
		\dispout & =\X \dispin + \dispch,\\ 
		\CMout & =\X \CMin \trans{\X} + \eff, 
	\end{split} 
\end{equation} 
If the input state is Gaussian, given by $\dispin$ and $\CMin$, then the output state of the channel is also a Gaussian state given by $\dispout$ and $\CMout$. In order to satisfy the conditions of a CPTP map $\Phi$ has to satisfy \cite{HW01}
\begin{equation}
	\eff+\frac{i}{2}\left(\bm{\Omega}-\X \bm{\Omega} \trans{\X}\right) \ge 0, 
\end{equation}
where $\bm{\Omega}$ is the symplectic matrix
\begin{equation}
	\bm{\Omega} = 
	\begin{pmatrix}
		0 & 1\\-1 & 0
	\end{pmatrix}.
\end{equation}
In the following, the parameters 
\begin{equation}\label{eq:tauy}
	\tau = \det{\X}, \quad y = \sqrt{\det{\eff}},
\end{equation}
will be of a particular interest because $\tau$ may be considered as a channel transmissivity (if $0 \le \tau\le1$) or amplification gain (if $\tau\ge1$), while $y$ characterizes the added noise.

Note that any single-mode covariance matrix $\CM$ can be diagonalized by a symplectic transformation corresponding to a rotation in the phase space which preserves the \emph{symplectic eigenvalue} $\se=\sqrt{\det{\CM}}$, as well as the
degree of squeezing of the state:
\begin{equation}\label{eq:genCM}
	\Sim\CM\Sim^{\rm T} =  \left(\M+\frac{1}{2}\right)
	\begin{pmatrix}
		 \fr^{-1} & 0\\
		0 & \fr
	\end{pmatrix}.
\end{equation}
Here the parameter $\fr$ accounts for the amount of squeezing of the state with CM $\CM$ and 
\begin{equation}\label{eq:Mth}
    \M\equiv\se-\frac{1}{2}=\sqrt{\det{\CM}}-\frac{1}{2}.
\end{equation}
The meaning of the coefficient $M$ will be clear if we compare it with the number of photons $N$
defined as the expectation value of the number operator $\hat N = \hat{a}^\dag \hat{a}$ which it takes in the state $\hat{\rho}$. 
The number of photons is related to the energy of the state given by the trace of its covariance matrix as follows
\begin{equation}\label{eq:Nphot}
  N\equiv \Tr[\hat{a}^\dag \hat{a}\hat{\rho}]=\frac{1}{2}(\Tr[\CM]-1).
\end{equation}
For a thermal state, i.e. $\fr=1$, equations \eqref{eq:Mth} and \eqref{eq:Nphot} imply that $\M = N$. Furthermore, let us notice that for arbitrary $\fr$ the covariance matrix $\CM$ in the diagonal form \eqref{eq:genCM} is identical (up to dimensional
factors made of $\hbar$) to the covariance matrix (expressed in natural units) of a thermal single-mode state of $\M$ photons and \textit{optical frequency} (see, for example, \cite{HSH99}). Inspired by
this similarity we call $\M$ the number of \emph{thermal photons} and $\fr$ the ``frequency'' of the state. Thus all pure states have zero temperature and zero thermal photons, $\M=0$.
Although in our case $\fr$ does not represent any optical frequency the similarity has a deep physical root. Indeed, given the number of thermal photons $\M$, the energy of the state is determined by both the optical frequency and the degree of squeezing.


The symplectic eigenvalue $\se$ and the number of thermal photons $\M$ of the covariance matrix $\CM$ of a single mode state are important parameters in what follows since they can be used to calculate the von Neumann
entropy of Gaussian states. Recall that the density matrix $\hat{\rho}$ of a thermal state of $M$ photons is diagonal in the Fock basis and therefore its von Neumann entropy
$S(\hat{\rho})=-\Tr[\hat{\rho}\log\hat{\rho}]$ may be calculated straightforwardly:
\begin{equation}\label{eq:VNE}
S(\hat{\rho})=g(\se-1/2)=g(\M),
\end{equation} 
where
\begin{equation}\label{eq:gx}
\begin{split}
	g(x) & =(x+1)\log_2(x+1)-x\log_2 x, \quad x>0.\\ 
	g(0) & = 0.
	\end{split}
\end{equation} 
Thus, the von Neumann entropy of a thermal state can be expressed as a function of the symplectic eigenvalue $\nu$ or the number of thermal photons $\M$. Then, any CM $\CM$ of the form \eqref{eq:genCM} can be transformed by a
symplectic (squeezing) transformation to the CM of a thermal state ($\fr=1$) with the same number of thermal photons $\M$ as the initial $\CM$. Noting that this transformation corresponds to a unitary
transformation of the corresponding density operator which preserves the von Neumann entropy, we recover a well known fact that the von Neumann entropy of any single-mode Gaussian state is a function of the number
of thermal photons $\M$ given by Eq.~\eqref{eq:VNE} 


\section{Gaussian capacity of the fiducial channel}\label{sec:capgauss}

The one-shot (single-letter) \emph{Gaussian capacity} is the maximal transmission rate of classical information when we are restricted to the use of Gaussian input states and Gaussian encodings such that different uses of the channel are not correlated. The energy constraint is always imposed on the input state because the unconstrained capacity is trivially infinite. The recently proven additivity
conjecture together with the proof of the minimum output entropy conjecture (see Refs.~\cite{GHG13,GGCH13}) assert that the Gaussian capacity is the classical capacity for phase-insensitive Gaussian channels (in case of squeezed Gaussian noise this holds 
if the input energy exceeds a certain threshold). 
However, in the absence of the full proof embracing all Gaussian channels and the entire input energy regime we distinguish the Gaussian capacity $\CapGchi(\chG)$ from the classical capacity $C_\chi(\chG)$. The former
reads~\cite{EW07,SKPPC13}:
\begin{equation}\label{eq:capgauss}
	\begin{split}
		\CapGchi(\chG,\Ninmod) & =\max_{\{\CMin,\CMinmod : \Tr[\CMinmod]=2\Ninmod+1\}}\chi_G,\\
		\chi_G & = g\left(\Moutmod\right)-g\left(\Mout\right).
	\end{split}
\end{equation}
where $\CMin$ is the CM of the input (letter) state and $\CMinmod$ is the CM of the average modulated input state with $\Ninmod$ photons. Note that we state in Eq.~\eqref{eq:capgauss} the explicit dependence of the Gaussian capacity on $\Ninmod$. 
The average (modulated) input state with CM $\CMinmod$ is realized by a Gaussian distribution of the displacements of pure Gaussian states with the same CM $\CMin$ (see, for example, \cite{SKPPC13}). That is why entropies which enter the general formula for the
capacity are expressed in terms of functions of thermal photons $\Mout$ and $\Moutmod$ corresponding to the \emph{output} and \emph{modulated output states} with CM $\CMout$ and $\CMoutmod$ given by
\begin{equation}\label{eq:CMoutCMoutmod}
	\CMout=\chG[\CMin], \quad \CMoutmod=\chG[\CMinmod].
\end{equation}
The input energy constraint is expressed in terms of the mean number of photons $\Ninmod$ in the modulated input state.

We showed in~\cite{SKPPC13} that the problem of calculating the Gaussian capacity of an arbitrary Gaussian channel $\chG$ can be reduced to calculating the capacity of a \emph{fiducial} Gaussian channel $\chPS_{(\tau,y,\frenv)}$, which is defined in our parametrization by
\begin{equation}\label{eq:chPS}
	\XPS = \begin{pmatrix}
		\sqrt{|\tau|} & 0\\0 & \sign(\tau) \sqrt{|\tau|} 
	\end{pmatrix}, \quad
	\effPS = y\begin{pmatrix}
		 \frenv^{-1} & 0\\
		0 & \frenv
	\end{pmatrix},
\end{equation}
where parameter $y$ is a function of $\tau$, namely, 
\begin{equation}\label{eq:ybytauMenv}
	y = \left\{ 
	  \begin{array}{ll}
		|1-\tau|\left(\Menv+\frac{1}{2}\right), & \quad \tau \ne 1,\\
		\Menv, & \quad \tau = 1.\\
	\end{array} 
		\right.
\end{equation}
Note that $\tau=1$ and $\Menv=0$ correspond to the perfect transmission channel. Otherwise, $0 \le \tau \le 1$ corresponds to a lossy channel, $\tau=1$ a classical additive noise channel, $\tau > 1$ an
amplification channel, and $\tau < 0$ a phase-conjugating channel. Although in the case of the classical additive noise channel the CM of the noise $\effPS$ defined in \eqref{eq:chPS} strictly speaking does not correspond 
to a valid ``state'' of the environment, we will loosely apply the term ``squeezing'' parameter to $\frenv$ for all channels.

The solution of the optimization problem \eqref{eq:capgauss} has already been obtained and intensively discussed for the lossy channel and classical additive noise channel (see~\cite{PZM08,LPM09,PLM09}
and~\cite{SDKC09,SKC10,SKC11}, respectively). In the present paper we generalize those results to the most general single-mode Gaussian channel.

In Eq.~\eqref{eq:capgauss} one has to maximize over the two covariance matrices $\CMin$ and $\CMinmod$ with the energy constraint $\Ninmod$ stated in Eq.~\eqref{eq:capgauss}, and the
constraint that the input state is pure, i.e. 
\begin{equation}\label{eq:purity}
	\det{\CMin}=\frac{1}{4}.	
\end{equation}
We can straightforwardly use Theorem 1 of~\cite{PLM09} (since the resulting equations are of the same form) to conclude that the optimal covariance
matrices $\CMin$ and $\CMinmod$ are diagonal in the same basis as the \emph{effective noise} CM $\effPS$. Therefore, in the following any covariance matrix $\CM$ will be considered to be diagonal. As a consequence,
the CM of the optimal output state and optimal modulated output state read
\begin{equation}\label{eq:CMoutdiag}
	\CMout = |\tau|\CMin + \effPS, \quad \CMoutmod = |\tau|\CMinmod + \effPS.
\end{equation}
where
\begin{equation}\label{eq:CMindiag}
	\CMin = \frac{1}{2}\begin{pmatrix}
		\frin^{-1} & 0\\0 & \frin
	\end{pmatrix},
	\CMinmod = \left(\Ninmod+\frac{1}{2}\right)\begin{pmatrix}
		\frinmod^{-1} & 0\\0 & \frinmod
	\end{pmatrix}.
\end{equation}
Now we present the solution to the optimization problem stated in Eq.~\eqref{eq:capgauss}. First, we derive a solution that will only be valid above an input energy threshold, and secondly, we present the solution for input energies below this threshold. 

\subsection{Solution above the input energy threshold}\label{sec:solat}
We start from an obvious bound on the Gaussian capacity, which follows from Eq.~\eqref{eq:capgauss} and the constraints discussed above
\begin{equation}\label{eq:capgbound}
	\CapGchi \le \max_{\{\CMinmod:\Tr[\CMinmod]=2\Ninmod+1\}} g(\Moutmod)  - \min_{\{\CMin:\det\CMin=1/4\}} g(\Mout).
\end{equation}
We would obtain the Gaussian capacity if we presented $\CMin$ and $\CMinmod$ which saturate this bound. This idea was used, for example, in \cite{H06}. However, by considering implications of the energy constraint in \eqref{eq:capgauss}, which $\CMinmod$ should respect, we will find the limitations of this method.
 
Let us start by finding the maximum of the first term of Eq.~\eqref{eq:capgbound} using the fact that the two terms are independent. Since $g(x)$ is a monotonously increasing function it is sufficient to
maximize its argument in order maximize its value. We can express the argument $\Moutmod$ in the first term of \eqref{eq:capgbound} in terms of the mean number of photons $\Noutmod$ of the modulated output state,
i.e.
\begin{equation}\label{eq:MoutmodNoutmod}
	\begin{split}
		\Moutmod & = \frac{2\Noutmod + 1}{\froutmod^{-1}+\froutmod} - \frac{1}{2},\\
		\Noutmod & = |\tau|\left(\Ninmod+\frac{1}{2}\right)+\frac{1}{2}\left(\Tr[\effPS] -1\right).
	\end{split}	
\end{equation}
Then, since $\Noutmod$ is fixed for a given number of input photons $\Ninmod$ and channel parameters $\tau$, $y$, and $\frenv$ (see Eq.~\eqref{eq:chPS}) the only free parameter in $\Moutmod$ is the frequency $\froutmod$ of the modulated output state. Its optimal value, which maximizes $\Moutmod$ in Eq.~\eqref{eq:MoutmodNoutmod}, is easily found:
\begin{equation}\label{eq:freqoutmodat}
	\froutmod=1, 
\end{equation}
which implies that 
\begin{equation}\label{eq:Moutmodat}
	\Moutmod = \Noutmod.
\end{equation}
This corresponds to the well-known fact that given a finite energy the thermal state maximizes the von Neumann entropy. 

Next we find the minimum of the second term on the right hand side of Eq.~\eqref{eq:capgbound}. In the same way as we maximized the first term by maximizing the argument $\Moutmod$, it is enough here to minimize
the argument $\Mout$. The number of thermal photons at the output $\Mout$ can be expressed in terms of the input frequency $\frin$ from definitions \eqref{eq:Mth} and \eqref{eq:CMoutdiag}, where $\frin$ is
the only free parameter.
Then it is straightforward to find that the value of $\Mout$ is minimized when the squeezing parameter of the input state is equal to the one of the noise
\begin{equation}\label{eq:resonance}
	\frin=\frenv.
\end{equation}
In terms of frequencies one may say that the input and environment modes are ``in resonance''. This also means that the optimal input state should have the same squeezing as the state of the environment.
Together with the uniform distribution of energy between the quadratures of the modulated output (thermal) state this is also a single-mode form of \emph{quantum water-filling}; a type of solution found also
in multimode problems, which corresponds to equal variances of all quadratures of all modes \cite{HSH99}, \cite{SDKC09,PLM09,SKC10,SKC11}.

Using Eq.~\eqref{eq:resonance} we can express $\Mout$ which corresponds to this solution in a simple form
\begin{equation}\label{eq:MoutWF}
	\Mout+\frac{1}{2}=\frac{|\tau|}{2}+y. 
\end{equation}
Together with Eqs.~\eqref{eq:chPS}-\eqref{eq:CMoutdiag} this implies furthermore, that the input and the output state are 
``in resonance'' as well, i.e.
\begin{equation}\label{eq:resonanceOut}
	\frin=\frout.
\end{equation}

The CM of the modulated input state can be expressed from the obtained results using the equation
\begin{equation}\label{eq:vmod}
	\CMoutmod-\CMout=|\tau|(\CMinmod-\CMin),
\end{equation} 
which directly follows from Eq.~\eqref{eq:CMoutdiag}. Then the number of thermal photons $\Minmod$ and the frequency $\frinmod$ are obtained following definitions \eqref{eq:genCM} and \eqref{eq:Mth} from the CM
$\CMinmod$ and the input energy constraint in the form
\begin{equation}\label{eq:Minmodbyfrinmod}
	\Minmod+\frac{1}{2}  = \frac{2\Ninmod+1}{\frinmod^{-1}+\frinmod}.
\end{equation}
After several steps we obtain
\begin{equation}\label{eq:frinmodat}
	\frinmod = \sqrt{\frac{|\tau|(2\Ninmod+1)+y(\frenv^{-1}-\frenv)}{|\tau|(2\Ninmod+1)-y(\frenv^{-1}-\frenv)}}
\end{equation}
and
\begin{equation}\label{eq:Minmodat}
	\Minmod + \frac{1}{2} = \frac{1}{2}\sqrt{(2\Ninmod+1)^2-\frac{y^2}{\tau^2}\left(\frenv^{-1} - \frenv\right)^2} .
\end{equation}

Here we should make an important note that the modulated state with CM $\CMinmod$ is obtained by averaging over displacements of the input letter state with CM $\CMin$ and therefore, the right hand side of
Eq.~\eqref{eq:vmod} should be a positive definite matrix. This requirement together with Eqs.~\eqref{eq:freqoutmodat} and \eqref{eq:resonance} can be jointly satisfied by some $\CMin$ and $\CMinmod$ only if the
input energy is above the following threshold:
\begin{equation}\label{eq:Nthr}
	\Nthr = \frac{1}{2\frenv}\left[ 1 + \frac{y}{|\tau|}\left|1 - \frenv^2\right| \right] - \frac{1}{2}.
\end{equation}
Hence, only for $\Ninmod \ge \Nthr$, the covariance matrices $\CMin$ and $\CMinmod$ found above determine the input letter state and modulated average input state which saturate the upper bound on the Gaussian
capacity \eqref{eq:capgbound}, and therefore realize the Gaussian capacity. This shows the limitation of the attempt to minimize and maximize terms in Eq.~\eqref{eq:capgauss} independently.

Using the results obtained above we can express in a closed form the Gaussian capacity of the fiducial channel as a function of the channel parameters and the input energy constraint: inserting into
Eq.~\eqref{eq:capgauss} the optimal values of $\Mout$ \eqref{eq:MoutWF} and $\Moutmod$ \eqref{eq:Moutmodat} where $\Noutmod$ is given by \eqref{eq:MoutmodNoutmod} one obtains
\begin{equation}\label{eq:capat}
	\begin{split}
		\CapGchi & = g\left(|\tau|\left(\Ninmod+\frac{1}{2}\right)+\frac{1}{2}\left(\Tr[\effPS] -1\right)\right)\\
		& -g\left(|\tau|\frac{1}{2}+y-\frac{1}{2}\right), \; \Ninmod \ge \Nthr.
	\end{split}	
\end{equation}
From Eqs.~\eqref{eq:frinmodat} and \eqref{eq:Minmodat} we can easily recover the well-known solution for thermal noise ($\frenv = 1$), which implies $\Minmod=\Ninmod$ and $\frin=\frinmod=1$ meaning that the optimal encoding corresponds to a displacement of coherent input states for so-called \emph{thermal} or \emph{phase-insensitive} channels (see \cite{SKPPC13}). 

Thanks to the recently proven minimum output entropy conjecture \cite{GHG13,GGCH13} Eq.~\eqref{eq:capat} provides not only the Gaussian capacity but the classical capacity for an arbitrary single-mode Gaussian channel $\chG$ as proven in \cite{SKPPC13,H14}, i.e. 
\begin{equation}\label{eq:CapGequalCap}
	\CapGchi(\chG) = C_\chi(\chG) = C(\chG), \; \Ninmod \ge \Nthr.
\end{equation}

\subsection{Solution below the input energy threshold}\label{sec:solbt}
Now we derive the solution for the case when the input energy is not enough to satisfy the water-filling solution. Up to now a solution for this case is known only for the lossy \cite{PLM09} and additive noise
channels \cite{SKC11}.

By definition of $\Nthr$ the condition $\Ninmod < \Nthr$ implies that Eqs.~\eqref{eq:freqoutmodat} and \eqref{eq:resonance} could only be satisfied if (at least) one of the diagonal terms of the difference
$\CMinmod - \CMin$ was negative. However, this would be unphysical because $\CMinmod$ is the CM of a state which is a mixture of states with CM $\CMin$. This means that the formal maximum of Eq.~\eqref{eq:capgauss} lays
outside the physically admissible region of optimization parameters and therefore the real (physically admissible) maximum lays in fact at the boundary of this region. This boundary is formed by the modulated input
CM $\CMinmod$ with (at least) one diagonal term being equal to the corresponding diagonal term of the non-modulated input CM $\CMin$. This implies that the optimal encoding only uses one of the two quadratures for information transmission (modulating neither of them is clearly not optimal). Without loss of generality we assume that the $p$-quadrature of the noise is squeezed and therefore $\frenv < 1$. Therefore, we chose the (more noisy) $q$-quadrature
to be no longer modulated:
\begin{equation}\label{eq:modqzero}
	\left(\Minmod+\frac{1}{2} \right)\frinmod^{-1}=\frac{1}{2}\frin^{-1}, 
\end{equation}
and argue below why this choice is optimal.

Since below the input energy threshold, there is no valid $\CMin$ and $\CMinmod$ which saturate independently the minimum and maximum of \eqref{eq:capgbound}, we have to maximize the whole
expression on the right hand side of Eq.~\eqref{eq:capgauss}. Therefor we use in the following the method of Lagrange multipliers.

Note that for the case $\Ninmod \ge \Nthr$ the optimization problem can be formulated in terms of the three parameters $\frin$, $\frinmod$ and $\Minmod$ satisfying the input energy constraint~\eqref{eq:Minmodbyfrinmod} [where furthermore, only pure input states are considered, see Eq.~\eqref{eq:purity}]. 

For the case $\Ninmod < \Nthr$ we can use the new constraint given by Eq.~\eqref{eq:modqzero} to express $\Minmod$ in terms of $\frin$ and $\frinmod$, which together with Eq.~\eqref{eq:Minmodbyfrinmod} leads to
\begin{equation}\label{eq:frininmod}
	\frin = \frac{1+\frinmod^2}{2(2\Ninmod+1)}.
\end{equation}
Together with the input energy constraint the number of independent variables would be reduced to one. However, we find it more instructive to use the energy constraint in the Lagrangian with two independent variables. We shall maximize Eq.~\eqref{eq:capgauss} with respect to the pair of variables $\frin$, $\froutmod$. In addition, we will replace the input energy constraint $\Ninmod$ by the output energy
constraint $\Noutmod$ which is equivalent to the original one for fixed channel parameters $(\tau,y,\frenv)$ as one can see in Eq.~\eqref{eq:MoutmodNoutmod}. Then, our Lagrangian is given by
\begin{eqnarray}\label{eq:lagrangian}
	{\mathcal L} & = & g(\Moutmod) - g(\Mout) \\ 
		&-&  \frac{\betaoutmod}{\ln 2}\left(\Tr[\CMoutmod] - \Noutmod - \frac{1}{2}\right).\nonumber
\end{eqnarray}
The choice of the Lagrange multiplier $\betaoutmod/\ln 2$ will be justified below. An extremum of the Lagrangian fulfills
\begin{equation}
	\nabla {\mathcal L} = 0,
\end{equation}
where the gradient reads
\begin{equation}
	\nabla = \trans{\left(\frac{\partial}{\partial \frin}, \frac{\partial}{\partial \froutmod}\right)}.
\end{equation}
Before taking the derivatives, recall the condition $\Ninmod<\Nthr$, which implies Eq.~\eqref{eq:modqzero} and therefore
\begin{equation}\label{eq:Moutmodbtv1}
	\Moutmod= \froutmod \left(\frac{|\tau|}{2\frin} + \frac{y}{\frenv} \right) - \frac{1}{2}.
\end{equation}
Furthermore, $\Mout$ can be expressed using the definition \eqref{eq:Mth}. Thus, we find
\begin{eqnarray}\label{eq:lagrsystem}
\frac{\partial {\mathcal L}}{\partial \frin} 
	& = & -
	\left[\left(2 g'(\Moutmod)\frac{\froutmod}{\frin}
	-\frac{\betaoutmod}{\ln 2}\left(\froutmod^{-1} + \froutmod\right)\right)\right.\nonumber\\
	& - & \left. g'(\Mout)\frac{\frout}{\frin}\left(1-\frac{\frin^2}{\frout^2}\right)\right]\frac{|\tau|}{4\frin},\label{eq:dLdfreqin}\\
	\frac{\partial {\mathcal L}}{\partial \froutmod} & = & \left(\Moutmod+\frac{1}{2}\right)\left(g'(\Moutmod)\froutmod^{-1} - \frac{\betaoutmod}{\ln 2}\right),\label{eq:dLdfreqoutmod}
\end{eqnarray}
where the first derivative of $g(x)$ is
\begin{equation}
  g'(x)=\frac{1}{\ln 2}\log_2\left(\frac{x+1}{x}\right), 
\end{equation}
and Eq.~\eqref{eq:dLdfreqin} was simplified using Eq.~\eqref{eq:CMoutdiag}. By equalizing the right hand side of Eq.~\eqref{eq:dLdfreqoutmod} to zero we recover the Bose-Einstein statistics for the amount of \emph{thermal photons} in the modulated output state
\begin{equation}\label{eq:Moutmodbt}
	\Moutmod = \frac{1}{e^{\froutmod \betaoutmod}-1},
\end{equation}
which means that the multiplier $\betaoutmod$ can be regarded as an \emph{inverse temperature} of the state. Note that Eq.~\eqref{eq:Moutmodbt} was already derived previously in by several authors \cite{YS80}, and in
particular, for the lossy channel with phase sensitive noise in~\cite{PLM09}. The physical meaning of the Lagrange multiplier $\betaoutmod$ as an inverse temperature (intensive parameter) becomes more meaningful in the case of an ``ensemble'', which in our case corresponds to a set of parallel single-mode channels with a priori different noises \cite{PLM09,SDKC09,SKC11,S13}.

We use the notion of temperature for the non-modulated output as well and in order to have unified notations we define the
function
\begin{equation}\label{eq:beta}
	\beta(M,\fr)  = \frac{g'(M)\ln 2}{\fr},
\end{equation}
by inverting Eq.~\eqref{eq:Moutmodbt}. In this way we obtain
\begin{equation}\label{eq:betabout}
 	\betaoutmod=\beta(\Moutmod,\froutmod), \qquad \betaout=\beta(\Mout,\frout).
\end{equation}
Then, using the definition $\betaout \equiv \beta(\Mout,\frout)$ we simplify Eq.~\eqref{eq:dLdfreqin} and thus, obtain an equation relating the non-modulated input frequency with the modulated and non-modulated
output frequencies and ``temperatures''.
\begin{equation}\label{eq:freqoutmodbt}
	 \betaoutmod( 1-\froutmod^2) = \betaout\left(\frin^2-\frout^2\right),
\end{equation}
This equation not only fully covers the interval $\Ninmod < \Nthr$, where it provides solution to the optimization problem, but is valid for arbitrary $\Ninmod$.

Indeed, we confirm that for $\Ninmod \ge \Nthr$ Eq.~\eqref{eq:freqoutmodbt} is satisfied by the solution which we obtained in Eqs.~\eqref{eq:freqoutmodat} and \eqref{eq:resonanceOut}, i.e. $\froutmod = 1$ and the ``resonance'' between the input and non-modulated output state, i.e. $\frin = \frout$. Recall that the latter is equivalent to the ``resonance'' between the input and the environment, i.e. $\frin = \frenv$.
 
The solution to Eq.~\eqref{eq:freqoutmodbt} gives the functional dependence of the optimal input frequency on the channels parameters and the amount of photons at the input: $\frin=\frin$($\tau$, $y$,
$\frenv$,$\Ninmod$). Then in its turn $\frin$ determines 
$\CapGchi$ in \eqref{eq:capgauss}. Following Eq.~\eqref{eq:CMoutdiag} and taking into account Eq.~\eqref{eq:modqzero} we obtain
\begin{eqnarray}
	\Moutmod + \frac{1}{2}
    	& = &\left(|\tau|\frac{\frin^{-1}}{2}+y\frenv^{-1}\right)^{1/2}\nonumber \\ \label{eq:Moutfrin}
    	& \times& \left(|\tau|\left(2\Ninmod+1-\frac{\frin^{-1}}{2}\right)+y\frenv\right)^{1/2}, \\
  	\Mout  +\frac{1}{2}
    	& = & \left(\left(|\tau|\frac{\frin^{-1}}{2}+y\frenv^{-1}\right)\left(|\tau|\frac{\frin}{2}+y\frenv\right)\right)^{1/2}, \nonumber
\end{eqnarray}
where in the second factor of $\Moutmod$ we also took into account the energy constraint in order to express both $\Mout$ and $\Moutmod$ as functions of only one variable $\frin$. Finally, the Gaussian capacity of
a Gaussian channel is found as
\begin{equation}\label{eq:CapChiOmeIn}
	\CapGchi\left(\Phi_{(\tau,y,\frenv)},\Ninmod\right) = \chi_G(\tau,y,\frenv,\Ninmod,\frin(\tau,y,\frenv,\Ninmod)),
\end{equation}
where $\chi_G(\tau,y,\frenv,\Ninmod,\frin)$ is expressed according to \eqref{eq:capgauss} with $\Mout$ and $\Moutmod$ given by \eqref{eq:Moutfrin} and $\frin(\tau,y,\frenv,\Ninmod)$ is a solution of
Eq.~\eqref{eq:freqoutmodbt} for given channel parameters $(\tau,y,\frenv,\Ninmod)$. 

Observing that Eq.~\eqref{eq:freqoutmodbt} is transcendental due to the logarithm in the definition of $\beta$ we conclude that one cannot obtain its solution in a closed analytic form below for $\Ninmod < \Nthr$. In order to solve the equation numerically we need to use the expressions of $\betaout$, $\betaoutmod$, $\frout$ and $\froutmod$ as functions of $\frin$, the channel parameters $(\tau,y,\frenv)$, and the input energy
constraint $\Ninmod$:
\begin{equation}\label{eq:frbt}
	\begin{split}
 		\frout & = \sqrt{\frac{|\tau| \frac{\frin}{2} + y \frenv}{|\tau|\frac{\frin^{-1}}{2}+y\frenv^{-1}}},\\
 		\froutmod & = \sqrt{\frac{|\tau|\left(2\Ninmod+1- \frac{\frin^{-1}}{2}\right) +  y \frenv}{|\tau|\frac{\frin^{-1}}{2}+y\frenv^{-1}}}.
	\end{split}
\end{equation}
We use this solution in order to discuss the dependencies of the Gaussian capacity on the channel parameters in detail in the next sections.

Before going to this discussion we argue now why the constraint Eq.~\eqref{eq:modqzero} is optimal for $\frenv < 1$. This statement was already proven in our previous work on the classical additive noise channel
($\tau=1$) with phase sensitive noise for the whole domain $\Ninmod < \Nthr$~\cite{SKC11}. Since the transcendental equation \eqref{eq:freqoutmodbt}, up to a change of variables, is identical to the corresponding transcendental equation in~\cite{SKC11} the optimality of Eq.~\eqref{eq:modqzero} is proven by Lemma 1 of~\cite{SKC11}.

\section{Role of channel parameters}\label{sec:chpara}
Now we are ready to analyze the capacity of the fiducial channel as function of the channel parameters.
\subsection{Dependence on transmissivity \& gain $\tau$}
We first discuss the dependency of the Gaussian capacity on the transmissivity/gain $\tau$, when the other channel parameters parameters $\Menv$, $\frenv$, and $\Ninmod$ are kept constant. We emphasize here that we keep the amount of thermal photons in the environment $\Menv$ constant and not $y$, so that the latter becomes a function of $\tau$ and $\Menv$ according to \eqref{eq:ybytauMenv}.

Let us consider first the limits $\tau \to \pm \infty$. Both limiting cases concern
only the amplification channel and phase-conjugating channel, hence
\begin{equation}\label{eq:limytau}
  \lim_{|\tau| \to \infty}\frac{y}{|\tau|} = \Menv+\frac{1}{2}.
\end{equation} 

Since $\chi_G$ [defined in Eq.~\eqref{eq:capgauss}] depends on $\Moutmod$ and $\Mout$ we first study their behavior. 
In the limit of large $|\tau|$, both arguments of $\chiG$, namely, $\Mout$ and $\Moutmod$ diverge, which follows straightforwardly from Eqs.~\eqref{eq:chPS}-\eqref{eq:CMoutdiag}.

Since $\lim_{x \to \infty}[g(x)-\log_2(x)]=1/\ln 2$ we can replace the function $g(x)$ in $\chi_G$ by $\log_2(x)$ and thus,
\begin{equation}\label{eq:chilim}
		\chiG_{(|\tau|\to \infty)} \equiv  \lim_{|\tau| \to \infty}\chiG = \lim_{|\tau| \to \infty}\log_2 \left(\frac{\Moutmod}{\Mout}\right).
\end{equation}
The input energy threshold stated in Eq.~\eqref{eq:Nthr} in this limit behaves as follows:
\begin{equation}\label{eq:Nthrlim}
	\lim_{|\tau| \to \infty}\Nthr  = \frac{1}{2\frenv}\left[ 1 + \left(\Menv+\frac{1}{2}\right)\left(1 - \frenv^2\right) \right] - \frac{1}{2}.
\end{equation}

For the input energies ``above the threshold'', i.e.\\$\Ninmod \ge \lim_{|\tau| \to \infty}\Nthr$, the limiting function $\chiG_{(|\tau|\to \infty)}$ is maximized by the quantum
water-filling solution and optimal squeezing of the input states stated in Eqs.~\eqref{eq:freqoutmodat} and \eqref{eq:resonance}. We express $\Mout$ from the definition
\eqref{eq:Mth} with the help of \eqref{eq:CMoutdiag}, and $\Moutmod$ from Eqs.~\eqref{eq:MoutmodNoutmod} and \eqref{eq:Moutmodat}. Using these expressions and Eq.~\eqref{eq:limytau} we obtain the Gaussian capacity, and as argued before Eq.~\eqref{eq:CapGequalCap}, the classical capacity $C$ in this limit:
\begin{equation}\label{eq:limCapGtauAT} 
	\lim_{|\tau|\to \infty}\CapGchi = \lim_{|\tau|\to
\infty}C = \log_2 \frac{\Ninmod+\Nenv+1}{\Menv+1} , \quad \Ninmod \ge \Nthr, 
\end{equation} 
where $\Nenv+1/2=(\Menv+1/2)(\frenv^{-1}+\frenv)/2$.

For input energies ``below the threshold'', i.e.\\$\Ninmod < \lim_{|\tau| \to \infty}\Nthr$, Eq.~\eqref{eq:modqzero} is satisfied and the input energy constraint implies
\begin{equation}
   \left(\Minmod+\frac{1}{2}\right)\frinmod = 2\Ninmod +1 - \frac{1}{2}\frin^{-1}.
\end{equation}
Hence, the limit \eqref{eq:chilim} becomes
\begin{equation}\label{eq:chiGlimbt}
	\chiG_{(|\tau|\to \infty)} = \frac{1}{2}\log_2 \left(\frac{2\Ninmod+1-\frac{1}{2}\frin^{-1} + \epp}{\frac{1}{2}\frin + \epp}\right), \quad \Ninmod < \Nthr,
\end{equation}
where we introduced the variable $\epp=(\Menv+1/2)\frenv$, i.e. the variance of the $p$-quadrature of the added environment state. Since $\log_2(x)$ is a monotonous function it is sufficient to maximize its argument thus maximizing the logarithm in \eqref{eq:chiGlimbt}. The maximum is given by the solution of the following equation with respect to $\frin$
\begin{equation}\label{eq:iqqlimtau}
	\frac{2\Ninmod+1-\frac{1}{2}\frin^{-1} + \epp}{\frac{1}{2}\frin + \epp} = \frin^{-2},
\end{equation}
which yields a unique positive root
\begin{equation}\label{eq:frinlimtau}
	\frintauinf = \left(\sqrt{1+\frac{2\Ninmod+1}{\epp}+\frac{1}{4\epp^2}}-\frac{1}{2\epp}\right)^{-1}.
\end{equation}
\begin{figure}
	\centering
		\includegraphics[trim=26 0 0 0 mm, clip=true]{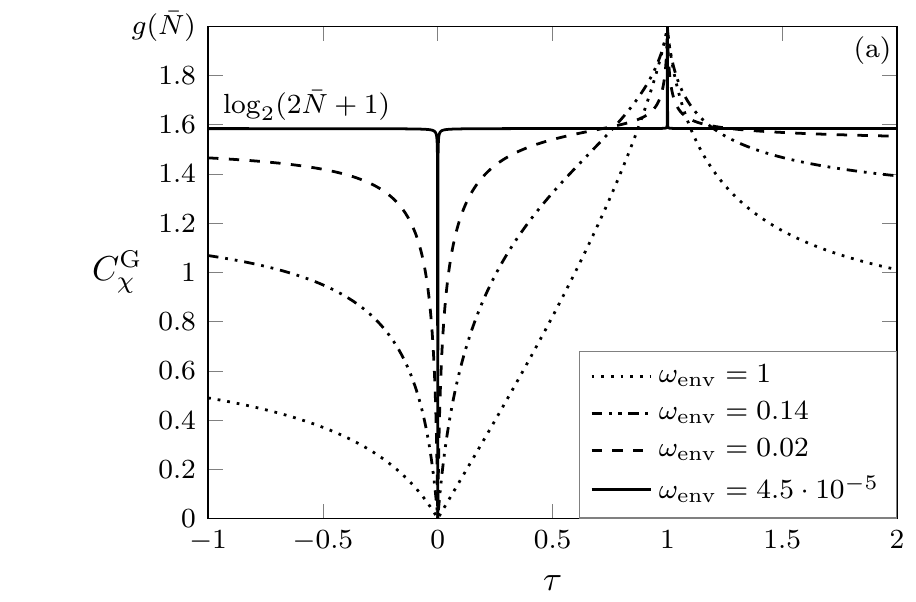}
		\includegraphics[trim=23 0 0 0 mm, clip=true]{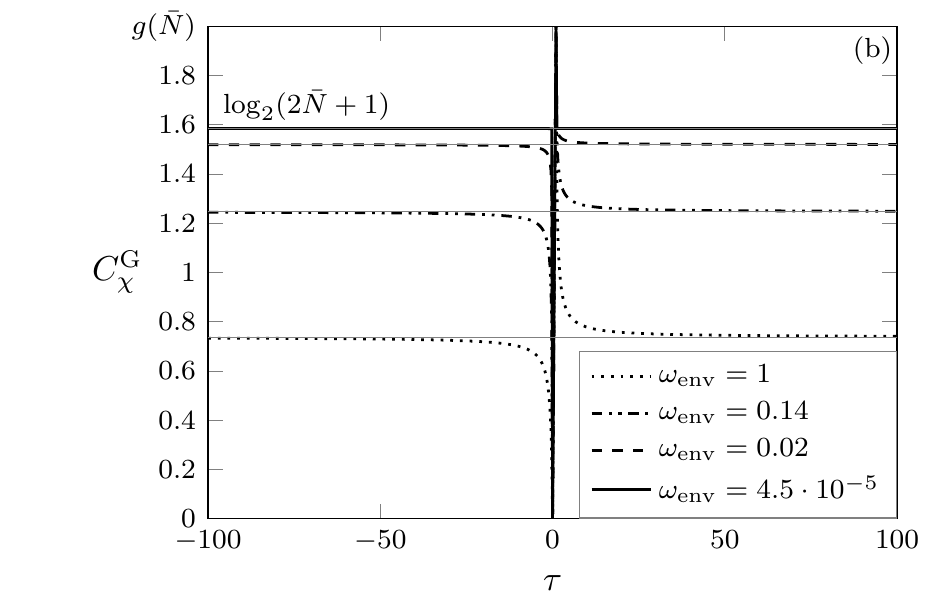}
	\caption{Gaussian capacity $\CapGchi$ vs. $\tau$. The different lines correspond to different values of the squeezing parameter $\frenv$ as indicated in the legend of the plots. Note that the point $\tau=1$ corresponds at the same time to the perfect transmission channel, the classical additive noise channel, or the limiting cases of the lossy channel and amplification channels depending on the choice of $y$, since $y$ is a function of $\tau$; here $\tau=1$ corresponds to the perfect transmission channel. The other parameters are $\Ninmod=1$ and $\Menv=0.1$. One clearly observes a common limiting behavior $\lim_{\frenv \to 0}\CapGchi = \log_2(2\Ninmod+1)$ of all channels (except for the perfect transmission channel and the zero transmission channel).}
	\label{fig:capvstau}
\end{figure}
Inserting \eqref{eq:iqqlimtau} into Eq.~\eqref{eq:chiGlimbt} we easily obtain
\begin{equation}\label{eq:limCapGtauBT}
	\lim_{|\tau|\to \infty}\CapGchi = -\log_2\left(\frintauinf\right), \quad \Ninmod < \Nthr,
\end{equation}
where $\frin$ is stated in Eq.~\eqref{eq:frinlimtau}. Note that the solution stated in \eqref{eq:limCapGtauBT} is identical to the solution of the \emph{homodyne rate} calculated in~\cite{PLM09}. (The homodyne rate is the maximal transmission rate when restricted to Gaussian encodings and homodyne detection.) 
Replacing $\epp$ by its expression in terms of $\Menv$ and $\frenv$ it is straightforward to show that
\begin{equation}
	\lim_{\frenv \to 0}\lim_{\tau \to \pm\infty}\CapGchi(\chPS,\Ninmod) = \log_2(2\Ninmod+1).
\end{equation}
We will come back to the discussion of this limit in the next subsection. 

Now we can straightforwardly prove the monotonicity of the capacities for the lossy channel ($\tau \in [0,1]$) and amplification channel ($\tau \ge 1$) using Lemma~\ref{lemPipeline} (see Appendix~\ref{ap:concatenated}). The classical capacity (as well as the Gaussian capacity) fulfills the pipelining property~\cite{NC00}
\begin{equation}\label{eq:PipeliningEquation}
	C(\chG_1 \circ \chG_2) \le \min\{C(\chG_1),C(\chG_2)\}.
\end{equation}
Then, using Lemma~\ref{lemPipeline} for $\tau_1,\tau_2 \ge 1$ (or $\tau_1,\tau_2 \in [0,1]$, or $\tau_1 < 0, \tau_2 \in [0,1]$) we obtain
\begin{eqnarray}
		\lefteqn{C\left(\chPS_{(\tau_2,y_2,s)} \circ \chPS_{(\tau_1,y_1,s)},\Ninmod\right)
		 =  C\left(\chPS_{(\tau_{\rm eff},y_{\mathrm{eff}},s)},\Ninmod\right)}\nonumber\\
		&  \le & \min \left\{ C\left(\chPS_{(\tau_1,y_1,s)},\Ninmod\right),  C\left(\chPS_{(\tau_2,y_2,s)},\Ninmod\right)\right\}.
\end{eqnarray}
where $y_i=|1-\tau_i|(\Menv+1/2)$, ($i=1,2$), the effective gain parameter is $\tau_{\rm eff} \equiv \tau_1\tau_2$  and the effective noise is characterized by $y_{\mathrm{eff}}= |1-\tau_{\rm eff}|(\Menv+1/2)$.

Amplification channels satisfy $\tau_{\rm eff} \ge \tau_1,\tau_2 \ge 1, \; \forall \tau_1,\tau_2$. An amplification channel with a higher gain
$\tau' > \tau \ge 1$ can always be decomposed into an amplification channel with gain $\tau_1=\tau$ followed by another amplification channel with
gain $\tau_2=\tau' / \tau$. As we see above, due to the pipelining property the classical capacity (and the Gaussian capacity) of the channel with
gain $\tau'$ cannot have increased. Similar relations for the lossy channel, for which the effective transmissivity satisfies $\tau_{\rm eff} \le
\tau_1,\tau_2$, $\forall \tau_1,\tau_2\in [0,1]$ lead to a conclusion that higher losses (lower transmissivity $\tau$) cannot increase the classical
capacity (and Gaussian capacity). In the case of the phase-conjugating channel ($\tau < 0$) Lemma~\ref{lemPipeline} states that a fiducial channel
with $\tau < \tau' < 0$ can always be decomposed into a phase-conjugating channel with $\tau_1 = \tau$ followed by a lossy channel with $\tau_2 =
\tau'/\tau$, where $\tau_2 \in [0,1]$. This implies that the capacity cannot have increased by increasing $\tau$.

The highest (Gaussian) capacity is given by the case $\tau=1$, i.e. the perfect transmission channel, for which $\CapG=C=g(\Ninmod)$. Since for $\tau\ge1$ increasing $\tau$ cannot
increase the Gaussian capacity $\CapGchi$, we proved that $\CapGchi$ monotonically tends from above to its limiting value stated in Eqs.~\eqref{eq:limCapGtauAT} and \eqref{eq:limCapGtauBT}. Thus,
\begin{equation}
	g(\bar{N}) \ge \CapGchi\left(\chPS_{(\tau\ge 1,y,s)},\Ninmod\right) \ge \lim_{\tau' \to \infty}\CapGchi\left(\chPS_{(\tau',y',s)},\Ninmod\right),
\end{equation}
where $y$ and $y'$ are functions of $\tau$, $\tau'$ and $\Menv$ according to \eqref{eq:ybytauMenv} and the latter is kept constant.

Equivalently, the lowest capacity is given by the case $\tau=0$, i.e. the zero-transmission channel. Since for $\tau \le 0$ decreasing $\tau$ cannot decrease $\CapGchi$, we proved that $\CapGchi$ monotonically tends from below to its limiting values stated in Eqs.~\eqref{eq:limCapGtauAT} and \eqref{eq:limCapGtauBT}. 

We illustrate the limiting behavior of $\CapGchi$ vs. $\tau$ for different values of squeezing $\frenv$ and the same $\Menv$ in Fig.~\ref{fig:capvstau}. One confirms the monotonic decrease of $\CapGchi$ for $\tau \ge 1$ and $\tau \le 0$. In addition, as argued above, on observes that indeed the limit $\lim_{\tau \to +\infty}\CapGchi$ is reached from above, whereas the limit $\lim_{\tau \to -\infty}\CapGchi$ is reached from below (with decreasing $\tau < 0$). 

\begin{figure}
	\centering
		\includegraphics[trim=12 0 0 2 mm, clip=true]{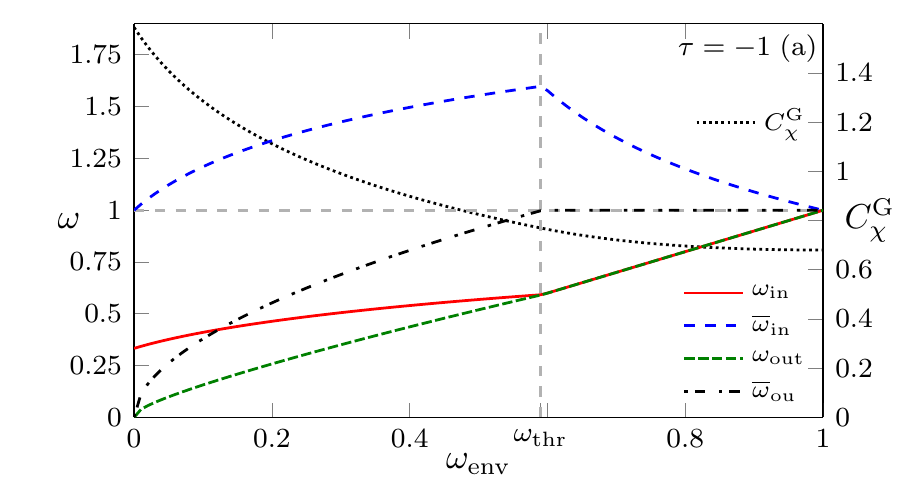}
		\includegraphics[trim=16 0 0 2 mm, clip=true]{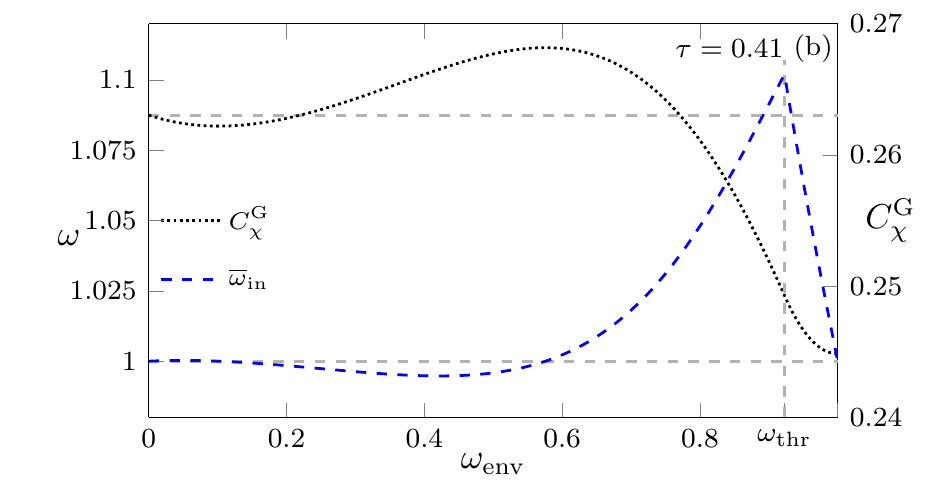}		
		\includegraphics[trim=12 0 0 2 mm, clip=true]{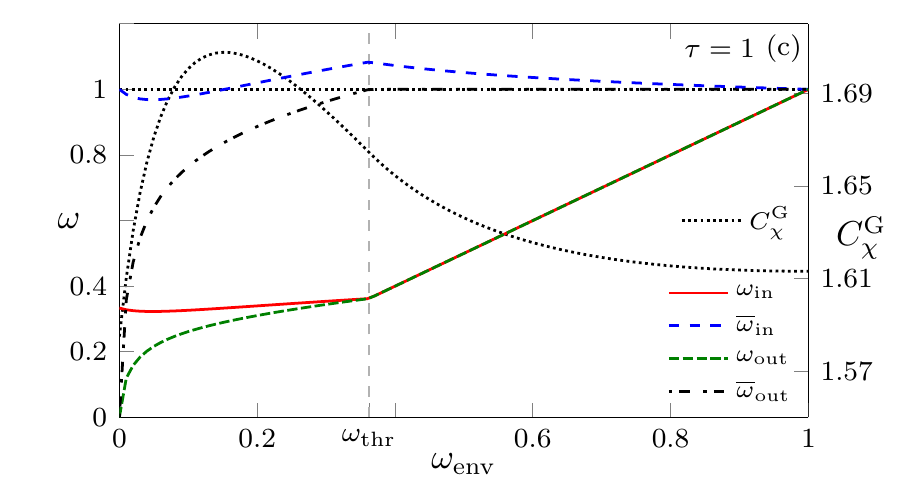}
		\includegraphics[trim=12 0 0 2 mm, clip=true]{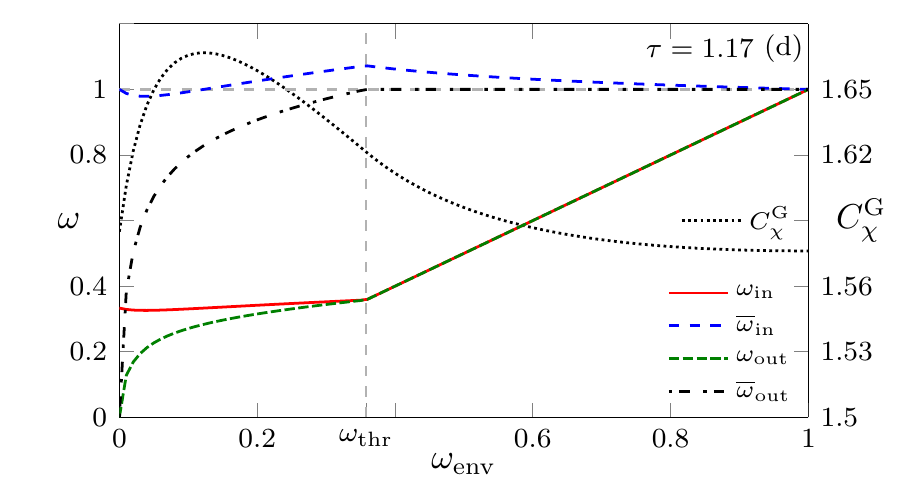}
	\caption{(Color online) Gaussian capacity $\CapGchi$ (right axis) and frequencies (left axis) of the input and output states, modulated and non-modulated, vs. $\frenv$ for different channels given by difference values of $\tau$, namely,
$\frin$ - upper solid (red) curve,
$\frinmod$ -  upper dashed (blue) curve,
 $\frout$ - lower densely-dashed (green) curve,
$\froutmod$ - dashed-dotted (black) curve,
 $\CapGchi$ - dotted (black) curve, right axis,
$\frthr$ - dashed vertical line.
The figures from top to bottom correspond to: 
phase-conjugating channel, $\tau=-1$, $\Ninmod=1$, $\Menv=0.1$, $\frthr=0.59$, 
lossy channel, $\tau=0.41$, $\Ninmod=0.1$, $\Menv=10^{-3}$, $\frthr=0.92$, 
additive noise channel $\tau=1$, $\Ninmod=1$, $\Menv=0.1$, $\frthr=0.362$, 
and amplifications channel $\tau=1.167$, $\Ninmod=1$, $\Menv=0.1$, $\frthr=0.358$. 
For the lossy channel ($\tau=0.41$) only $\frinmod$ (dashed blue) and $\CapGchi$ (dotted black) are plotted.
}
	\label{fig:freqvsfreq}
\end{figure}


\subsection{Dependence on noise squeezing $\frenv$}
Now we analyze the behavior of the Gaussian capacity as a function of the noise squeezing, while keeping the other channel parameters $\tau$, $y$, and $\Ninmod$ constant. Note that in this case both $y$ and
$\Menv$ are constant because $\tau$ does not change. We first investigate $\CapGchi$ in the limit of infinite squeezing $\frenv \to 0$. Since (as discussed above) the only arguments of $\chi_G$ are $\Moutmod$ and $\Mout$ [see Eq.~\eqref{eq:capgauss}] we first study their behavior in this limit. From Eqs.~\eqref{eq:Mth}, \eqref{eq:CMoutCMoutmod}, and \eqref{eq:chPS} it follows that
\begin{equation}\label{eq:Moutinf}
	\begin{split}
		\lim_{\frenv \to 0}\Mout & = \lim_{\frenv \to 0}\sqrt{\frac{|\tau|y\, \frin}{\frenv} \frac{1}{2}} \to \infty,\\
		\lim_{\frenv \to 0}\Moutmod & = \lim_{\frenv \to 0}\sqrt{\frac{|\tau|y\, \frinmod}{\frenv} \left(\Ninmod+\displaystyle\frac{1}{2}\right)} \to \infty.
	\end{split}	
\end{equation}

This implies that the difference of $g$-functions can again be replaced by a difference of logarithms, i.e.
\begin{equation}\label{eq:limchiG}
	\begin{split}
		\lim_{\frenv \to 0}\chi_G & = \lim_{\frenv \to 0}\log_2\left(\frac{\Moutmod}{\Mout}\right)\\
		& =\lim_{\frenv \to 0}\frac{1}{2} \log_2\left((1+2\Ninmod) \frac{ \frinmod}{\frin}\right)\\
									& = \frac{1}{2}\log_2(1+2\Ninmod)\\
									& +\frac{1}{4}\log_2\left(\frac{(2+4\Ninmod)\frin-1}{\frin^2}\right),
	\end{split}	
\end{equation}
where at the last step we used Eq.~\eqref{eq:frininmod}. Now, as in the case $\tau \to \infty$, it is sufficient to maximize the argument of the logarithm in order to obtain the Gaussian capacity in this limit. The argument is maximized by input states with squeezing
\begin{equation}\label{eq:limfreqin}
	\frininf = (1+2\Ninmod)^{-1},
\end{equation}
which leads to
\begin{equation}\label{eq:limcapgauss}
	\lim_{\frenv \to 0}\CapGchi = \log_2(1+2\Ninmod)=-\log_2\left(\frininf\right).
\end{equation}
Interestingly, Eq.~\eqref{eq:limfreqin} together with Eq.~\eqref{eq:modqzero} imply that the corresponding optimal modulated input state is a thermal state, i.e.
\begin{equation}\label{eq:limfreqinmod}
	\frinmod = 1, \quad \Minmod=\Ninmod.
\end{equation}
Note that the limit stated in Eq.~\eqref{eq:limcapgauss} was already derived in~\cite{PLM09} for the phase-sensitive lossy channel and here we confirm that it is valid for an arbitrary single-mode Gaussian channel as well.

This result can be interpreted as follows: In the limit of infinite squeezing of the
noise, $\frenv \to 0$, one of the noise quadratures vanishes (here: the $p$-quadrature)
and the other becomes infinite. Then, only one degree of freedom is available for
information transmission. At the same time it no longer suffers from the noise induced
by the channel but only from the quantum noise induced by the input state itself.
We observe that the one-shot Gaussian capacity then coincides with the doubled Shannon
capacity $C_{\rm Sh} = \frac{1}{2}\log_2(1+SNR)$ of a classical channel with
signal-to-noise ratio $SNR=\Ninmod/(1/2)$. The limiting value stated in
Eq.~\eqref{eq:limcapgauss} is confirmed in Fig.~\ref{fig:capvstau} where the Gaussian
capacity of the fiducial channels is plotted as a function of $\tau$ for different
$\frenv$.

Let us analyze next how the Gaussian capacity depends on the environment frequency
in the whole interval $0 < \frenv \le 1$. Examples of these curves for different types of
channels are presented in Fig.~\ref{fig:freqvsfreq}. One can see that the Gaussian
capacity, the optimal modulation and therefore, the output frequencies can be non-monotonous functions of $\frenv$
with maxima and/or minima. Generalizing several results that were previously obtained
for the lossy channel in~\cite{PLM09} we will study the extremal points of the capacity
curve at a finite noise squeezing $0 < \frenv \le 1$.

We begin with the input energy domain $\Ninmod \ge \Nthr$. Given $\Ninmod$ and $y$ we first obtain the \emph{threshold noise frequency} $\frthr$ from $\Nthr$ [stated in Eq.~\eqref{eq:Nthr}]. Although the equation has two roots, there is always only one positive
\begin{equation}
	\frthr = \frac{\displaystyle1+\frac{y}{|\tau|}}
	{\sqrt{\displaystyle\left(\frac{y}{|\tau|}\right)^2+\frac{y}{|\tau|} + \left(\Ninmod+\frac{1}{2}\right)^2}
	+\Ninmod+\frac{1}{2}},
\end{equation}
which is a decreasing function of $\Ninmod$ taking value $\frthr =1$ for $\Ninmod = 0$ so that for all $0 < \Ninmod < \infty$ we have $0 < \frthr < 1$. For $\frthr \le \frenv \le 1$ the solution for the capacity stated in Eq.~\eqref{eq:capat} holds, from which we conclude that
\begin{equation}
	\frac{\partial \CapGchi}{\partial \frenv} \le 0,  \quad   \frthr \le \frenv \le 1, 
\end{equation}
where equality is realized only if $\frenv = 1$, and furthermore,
\begin{equation}
	\frac{\partial^2 \CapGchi}{\partial \frenv^2} \ge 0, \quad \frthr \le \frenv \le 1,
\end{equation}
where also equality is only realized $\frenv = 1$. Hence, for $\frenv \ge \frthr$, the Gaussian capacity is a monotonously
decreasing and convex function with respect to $\frenv$ attaining its minimum at the boundary $\frenv = 1$ of this interval. This
means that although squeezing adds more energy to the noise the thermal noise is worst for the capacity in this interval of
$\frenv$. This observation was made first for the lossy channel in~\cite{PLM09} and for the additive noise
channel in~\cite{SKC10}. Finally, we conclude that the maximum of the Gaussian capacity as a function of $\frenv$ lays in the
interval $\frenv\in [0,\frthr]$, where, as we have shown above, the capacity is found by the solution ``below the threshold''. In this
case there is no closed formula for the Gaussian capacity. In order to find a maximum of the Gaussian capacity in this interval
of $\frenv$ one can study an implicit solution given by equations \eqref{eq:Moutmodbt} and \eqref{eq:freqoutmodbt}. However, we will find this maximum using another method. 

We will consider from the beginning a modified optimization problem with the same Lagrangian \eqref{eq:lagrangian} but taking $\frenv$ as an additional degree of freedom, so that the gradient reads
\begin{equation}
	\nabla = \trans{\left(\frac{\partial}{\partial \frin}, \frac{\partial}{\partial \froutmod}, \frac{\partial}{\partial \frenv}\right)}.
\end{equation}
Note that now the output energy $\Noutmod$ is no longer constant. It is a function of $\frenv$ [see Eq.~\eqref{eq:MoutmodNoutmod}] because
$\Tr[\effPS]$ is changing with $\frenv$. We have need to take into account this functional dependence when taking the derivative of $\CapGchi$ with
respect to $\frenv$. As a result, in addition to the two equations stated in Eq.~\eqref{eq:lagrsystem}, we have the third equation
\begin{equation}\label{eq:lagrange3}
       \frac{\partial {\mathcal L}}{\partial \frenv} = 0,
\end{equation} 
which can be simplified with the help of Eqs.~\eqref{eq:Moutmodbtv1}, \eqref{eq:Moutmodbt}--\eqref{eq:betabout}: 
\begin{equation}\label{eq:freqoutmodwenv}
	\betaoutmod\left(\frenv^2 -\froutmod^2\right)=  \betaout\left(\frenv^2-\frout^2\right).
\end{equation}
The trivial solution of the latter is given when both sides of the equation vanish. This is the case of thermal noise, i.e. $\frenv=1$, for which $\froutmod = \frout = 1$ is optimal. Note that $\frenv=1$ corresponds to a (local) minimum of the Gaussian capacity, and not a maximum [see Fig.~\ref{fig:capsvswenv}].

A joint solution of \eqref{eq:Moutmodbt}, \eqref{eq:freqoutmodbt}, and \eqref{eq:freqoutmodwenv} gives us both the optimal noise and optimal input frequency, $\frenv$ and $\frin$, respectively, corresponding to the extrema of the Lagrangian. In order to find the solution we express $\betaout/\betaoutmod$ from Eq.~\eqref{eq:freqoutmodbt} and insert the result into Eq.~\eqref{eq:freqoutmodwenv} thus obtaining
\begin{equation}\label{eq:freqenvmax}
	\frac{\froutmod^2-1}{\froutmod^2-\frenv^2} = \frac{\frout^2-\frin^2}{\frout^2-\frenv^2}.
\end{equation}
Injecting into this equation the definition for $\frout$, stated in Eq.~\eqref{eq:frbt}, we simplify its right hand side to
\begin{equation}\label{eq:freqoutmodenvrhs}
	\frac{\frout^2-\frin^2}{\frout^2-\frenv^2} = -\frac{2y\frin}{|\tau| \frenv}.
\end{equation}
Then using the expression for $\froutmod$, given by Eq.~\eqref{eq:frbt}, on the left hand side of \eqref{eq:freqenvmax} and equalizing it with the right hand side of Eq.~\eqref{eq:freqoutmodenvrhs} we obtain
\begin{equation}
	\frin=\frininf.
\end{equation}
Interestingly, this solution is equal to the optimal input frequency obtained for the limit of infinite squeezing, stated in Eq. \eqref{eq:limfreqin}. Note that the latter again implies that the optimal modulated input state is a thermal stated, i.e. $\frinmod=1$ [see Eq.~\eqref{eq:limfreqinmod}]. 

Now, in order to obtain the optimal noise frequency, we insert solutions $\frin=\frininf$ and $\frinmod=1$ into  Eq.~\eqref{eq:freqoutmodwenv} which leads to
\begin{equation}\label{eq:frenvext}
	\left(\Moutmod+\frac{1}{2} \right)\left(\froutmod - \frac{1}{\froutmod} \right) = y \left( \frenv - \frac{1}{\frenv}\right),
\end{equation}
which states that the difference between the variances of the modulated output state is exactly matched by the difference between the noise variances.


In order to find numerically the solution of Eq.~\eqref{eq:freqoutmodwenv} we express $\frout$  in the form \eqref{eq:frbt}, $\Mout$ in the form \eqref{eq:Mth} taking into account \eqref{eq:CMoutdiag}, where we inject $\frin=\frininf$.
Also we use similar expressions for $\froutmod$ and $\Moutmod$:
\begin{eqnarray}
   \froutmod & = & \sqrt{\left(\displaystyle \frac{|\tau|}{2\frininf}+y\frenv\right)\bigg/\left(\displaystyle\frac{|\tau|}{2\frininf}+y\frenv^{-1}\right)},\label{eq:froutmotext}\\
	\Moutmod & = & \sqrt{\left(\displaystyle \frac{|\tau|}{2\frininf}+y\frenv\right)
				\left(\displaystyle\frac{|\tau|}{2\frininf}+y\frenv^{-1}\right)}  - \frac{1}{2},\label{eq:Moutmod}\nonumber
\end{eqnarray}
where we have taken into account  Eq.~\eqref{eq:limfreqinmod}.

If the solution of the transcendental Eq.~\eqref{eq:freqoutmodwenv} together with Eqs.~\eqref{eq:limfreqin} and \eqref{eq:limfreqinmod} exists, it provides the extremal points of $\frenv$ corresponding to the local maxima or minima of $\CapGchi$, which is equal to $\chi_G$ [see Eq.~\eqref{eq:capgauss}] evaluated at these points.

\begin{figure}
	\centering
		\includegraphics[trim=5 0 0 0 mm, clip=true]{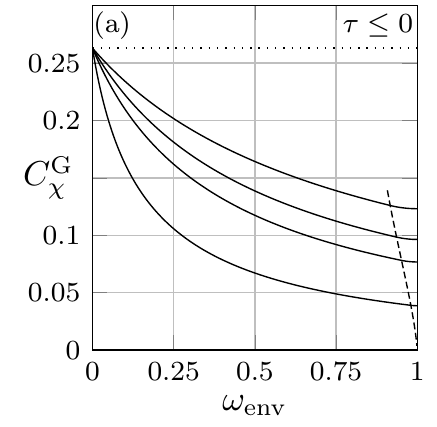}
		\includegraphics[trim=5 0 0 0 mm, clip=true]{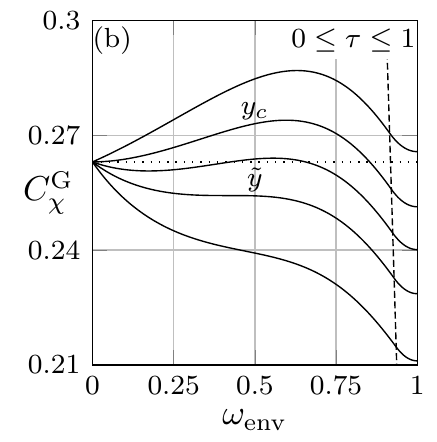}		
		\includegraphics[trim=5 0 0 0 mm, clip=true]{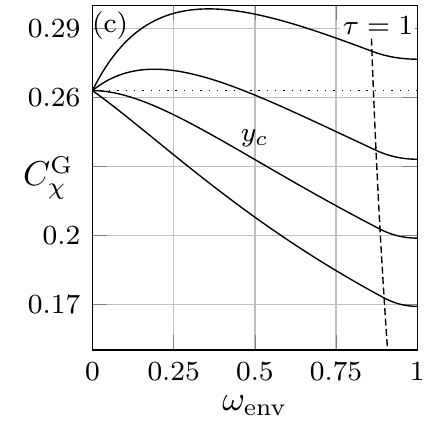}
		\includegraphics[trim=5 0 0 0 mm, clip=true]{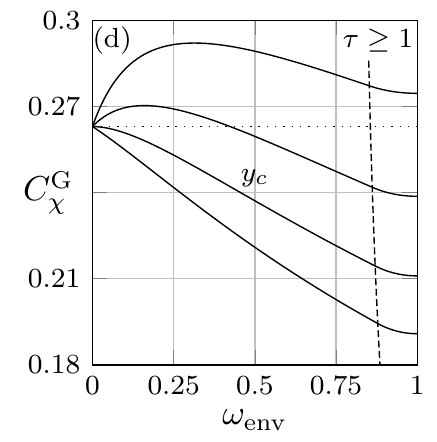}
	\caption{Capacities $\CapGchi$ vs. environment frequency $\frenv$ (from bottom to top) : 
	(a)  phase-conjugating ch.:  $\tau=\{-0.125, -0.5, -1, -4\}$, 
	(b):  lossy ch.: $\tau=\{0.34,\, 0.3759,\,0.41,\,0.455\}$, 
	(c):  additive noise ch.:  $y=\{0.125,\,0.2,\,1/\sqrt{12},\,0.4\}$
	(d): amplification ch.: $\tau = \{1.2,1.35,1+2/\sqrt{12},1.9\}$.
	For all plots we chose $\Ninmod=0.1$ and for plots (a), (b) and (d) $\Menv=10^{-3}$. For $\tau \in [0,1]$ the curve corresponding to $y=\tilde{y}$ has exactly a saddle point at finite squeezing, two extrema can be found for $\tilde{y} < y < y_c$ and one maximum is found for $y<y_c$. For the classical additive noise channel ($\tau=1$) and the amplifier ($\tau \ge 1$) the curves corresponding to $y \ge y_c$ are monotonically decreasing with $\frenv$ whereas for $y<y_c$ a maximum appears. The dashed straight line separates the solutions ``below'' (left) and ``above'' (right) the threshold $\Nthr$.}
	\label{fig:capsvswenv}
\end{figure}
\subsection{Classification}
In the following we classify the Gaussian capacity $\CapGchi$ with respect to the noise frequency $\frenv$. Our numerical analysis (see some typical graphs in Fig.~\ref{fig:capsvswenv}) revealed that all curves have no more than two inflection points and only the following scenarios in the open interval $\frenv \in (0,1)$ exist:
\begin{enumerate}[label=(\roman*)]
	\item\label{itm:noExt} there are no extremal points and the capacity is a monotonous function of $\frenv$,
	\item\label{itm:oneMax} there is only one extremal point corresponding to a maximum, 
	\item\label{itm:saddle} there is only one extremal point which is a saddle point,
	\item\label{itm:maxMin} there are two extrema: one minimum and one maximum. 
\end{enumerate}
All four different scenarios are observed only for the lossy channel (see also~\cite{PLM09}), whereas for the other channels either no extrema or only one maximum were found.
Scenarios \ref{itm:noExt} and \ref{itm:oneMax} were observed and studied for the classical additive noise channel in~\cite{SKC10}.

Although we do not have an analytic proof of these observations we take the statements above for granted in order to deduce a classification of the dependency of $\CapGchi$ on $\frenv$. We studied a wide range of parameters and did not find any curve, which violates our classification, and that is why we conjecture our hypothesis to be true.

First, we exploit the topological argument based on the observations
above and the behavior of the capacity in the limits $\frenv \to 0$ and $\frenv \to 1$. The last limit
was discussed in the previous subsection with the conclusion that for $\frenv\rightarrow1$ the capacity
is a monotonically decreasing and convex function of $\frenv$.

A simple topological argument allows us to conclude that if at $\frenv = 0$ the function $\CapGchi(\frenv)$ is
\begin{itemize}
\item[-]
\emph{increasing}: scenario \ref{itm:oneMax} can be realized (one maximum) with one inflection point if $\CapGchi$ is concave at $\frenv = 0$ and with two inflection points if it is convex. All other scenarios cannot be realized since they would require more than two inflection points, which was never observed and as we conjecture never happens.

\item[-]
\emph{decreasing and concave}: only scenario \ref{itm:noExt} can be realized, i.e. $\CapGchi$ is monotonously decreasing in the whole interval of $\frenv$ with one inflection point. Any other scenario would require more than two inflection points.

\item[-]
\emph{decreasing and convex}: all possible scenarios except \ref{itm:oneMax} can be realized.
\end{itemize}

Let us find now the domains of the channel parameters where the first and the second derivatives of $\CapGchi$ over $\frenv$ have definite signs at $\frenv \to 0$. We will study the coefficients of the expansion of $\CapGchi$ given by Eq.~\eqref{eq:CapChiOmeIn} in powers of $\frenv$ around the point $\frenv=0$.
\begin{eqnarray}\label{eq:expansionCapGchi}	
	\lefteqn{\CapGchi(\tau,y,\frenv,\Ninmod)  }  \\
		& \approx  &  -\log_2(\frininf)  +  a \, \frenv + b \, \frenv^2 + c \, \frenv^3
		 + {\mathcal{O}(\frenv^4)}, \nonumber
\end{eqnarray}
where we used our result on the optimal encoding $\frin=\frininf$ at $\frenv=0$ [see Eq.~\eqref{eq:limfreqin}]. Note that the neighborhood of $\frenv = 0$ lays in the domain $\frenv \in (0,\frthr)$ which corresponds to the solution for $\Ninmod < \Nthr$. Since no explicit expression for $\frin(\tau,y,\frenv,\Ninmod)$ exists in this domain we obtain the
coefficients on the right hand side of Eq.~\eqref{eq:expansionCapGchi} using an the expansion of $\frin$ in powers of $\frenv$:
\begin{equation}\label{eq:frinTaylor}
  \frin \approx \frininf + \alpha\frenv + \beta\frenv^2+\frac{\gamma}{6}\frenv^3 +\mathcal{O}(\frenv^4).
\end{equation}
The unknown coefficients $\alpha$, $\beta$, and $\gamma$ are found by expanding the transcendental
Eq.~\eqref{eq:freqoutmodbt} in powers of $\frenv$ with the help of \eqref{eq:frinTaylor}. In this way we
obtain a set of equations, each corresponding to a particular power of $\frenv$. The equations
corresponding to the first and second power provide us with the solutions:
\begin{eqnarray}\label{eq:alphabeta}
  \alpha & = & -\ln2K_1\frininf\left(y^2-y_c^2\right),\\
  \beta  & = &  -\ln2K_2\frininf\
  		\left[\frac{15}{2}\tilde{\tau}_c^2\left(y^2-y_c^2\right)^2
  		- y_c^2\left(\tau^2-\tilde{\tau}_c^2\right) \right],\nonumber\\
  K_j & =&  \frac{(-1)^j}{\ln 2 }\frac{1-\frininf^{2}}{\left(\frininf|\tau|y\right)^j}\nonumber
\end{eqnarray}
where the new notations
\begin{eqnarray}\label{eq:yc}
	y_c^2 & = & \frac{1}{12}, \\
	\label{eq:ttc}
	\tilde{\tau}_c^2 & = & \frac{2}{15}\left(1+\frininf^{2}\right),\nonumber
\end{eqnarray}
correspond, as we will show below, to the critical parameters determining the domains we are looking for.

Then the coefficients of the expansion \eqref{eq:expansionCapGchi} are found in the form:
\begin{eqnarray}
\label{eq:coefa}	
	a & = &
	 K_1\left(y^2-y_c^2\right),  \\
\label{eq:coefb}
	b & = & 
		b' + K_2\left(y^2-y_c^2\right)^2\left(1-\frac{15}{4}\tilde{\tau}_c^2\right), \\
\label{eq:coefbp}
	b' & = & K_2 \left[\frac{15}{2}\tilde{\tau}_c^2\left(y^2-y_c^2\right)^2
			-\frac{1}{2}y_c^2\left(\tau^2- \tilde{\tau}_c^2\right)\right], \\
\label{eq:coefc}
	c & = & c' + K_3 \left(2-\frac{15}{2}\tilde{\tau}_c^2\right)\left(y^2-y_c^2\right)\\
	& \times & \left[\left(1+\frac{15}{2}\tilde{\tau}_c^2\right)\left(y^2-y_c^2\right)^2
	-y_c^2\left(\tau^2- \tilde{\tau}_c^2\right)\right]\\
\label{eq:coefcp}
	c' & = & K_3 \left[\frac{4}{3}\left(y^2-y_c^2\right)^3
			\left(\left(1-\frac{15}{2}\tilde{\tau}_c^2\right)^2+\frac{15}{2}\tilde{\tau}_c^2\right)\right.  \\ 
	   & - & \frac{15}{2}\left(y^2-y_c^2\right)y_c^2\tilde{\tau}_c^2
	   	\left(\tau^2- 2\tilde{\tau}_c^2+\frac{4}{15}\left(1-\frac{2}{15}\tilde{\tau}_c^{-2}\right)\right)\nonumber\\
	  & - & \left.\frac{1}{48}\left(\left(\tau^2-\tilde{\tau}_c^2\right)^2
			-\frac{1}{21}\left(\left(\tilde{\tau}_c^2+\frac{4}{3}\right)^2-\frac{32}{15}\right)\right)\right].\nonumber
\end{eqnarray}

This expansion shows that in the limit $\frenv\rightarrow0$ the partial derivatives of $\CapGchi$ with respect to $\frenv$ up to the third order coincide with the full derivatives for 
 \begin{equation}\label{eq:ycrit}
	y=y_c \equiv \frac{1}{\sqrt{12}}.
\end{equation}
Moreover, if the latter is satisfied, the dominant linear term in the expansion
\eqref{eq:expansionCapGchi} is canceled because the coefficient $a$ vanishes. As we discussed above, this
would imply that $\frenv=0$ is an extremal point of $\CapGchi$ for $y=y_c$ and therefore, the value of $y$ alone decides how the capacity approaches the limiting value $\log_2(1+2\Ninmod)$, i.e., from below or above. For $y < y_c$ we find that $\rm{d} \CapGchi / \rm{d}\frenv > 0$ in the vicinity of $\frenv=0$. Then from the topological arguments discussed above we conclude that scenario \ref{itm:oneMax} is realized and $\CapGchi$ has one local maximum laying in the interval $0 < \frenv < \frthr$ . Thus, the line $y=y_c$ clearly separates the domain of channel parameters where only scenario \ref{itm:oneMax} is realized from the domain where only other scenarios are possible (see Fig.~\ref{fig:ytaudivision}).

Let us analyze the conditions on the other channel parameters which allow $y$ to attain the critical value. We notice first, that the condition $y=y_c$ can only be fulfilled when $\tau$ lays in a particular interval.
Indeed, for pure noise ($\Menv=0$) we obtain from \eqref{eq:ybytauMenv} that
\begin{equation}\label{eq:ypure}
	y = \frac{|1-\tau|}{2},
\end{equation}
which provides a lower bound on the physical region of $y$ (see Fig.~\ref{fig:ytaudivision}). In particular, for $\tau=1$ the lower bound is $y=0$ which is a valid lower bound as for the classical additive noise channel as well as for limiting cases of the lossy and amplification channels. Then, for an arbitrary Gaussian noise with parameter $y$ given by Eq.~\eqref{eq:ybytauMenv}, equation $y=y_c$ has two solutions:
\begin{equation}\label{eq:tauc}
	\tau_c = 1 \pm \frac{1}{\sqrt{3}(2\Menv+1)}.
\end{equation}
 Both solutions belong to the interval bounded by the values of $\tau_c$ for pure noise ($\Menv=0$) :
\begin{equation}
	\tau_L \le \tau_c \le \tau_R,\qquad \tau_{R,L} \equiv 1\pm\frac{1}{\sqrt{3}}.
\end{equation}
If $\tau$ lays outside this interval then the parameter $y\le y_c$ can never be attained by any $\Menv$ and therefore, scenario \ref{itm:oneMax} (one maximum) can never be realized. In Fig.~\ref{fig:ytaudivision} this corresponds to the physical region above the dark gray triangle.

We further observe that in the region defined by $y > y_c$ the first derivative of $\CapGchi$ at both
edges of the interval $\frenv \in [0,1)$ is negative. It implies
that although scenario \ref{itm:oneMax} is excluded, all three other scenarios are possible. Namely, the
function $\CapGchi(\frenv)$ may be either monotonously decreasing or have one maximum and one minimum
[both located in the interval $\frenv \in (0,\frthr)$] or have a saddle point, see
Fig.~\ref{fig:capsvswenv}.

In order to discriminate the number of extrema further we need to study the second derivative determining the factor $b$ of the next term of Taylor expansion \eqref{eq:expansionCapGchi}, which is quadratic in
$\frenv$. Note that contrary to $a$, factor $b$ [see Eq.~\eqref{eq:coefb}] is not always zero for the critical value $y=y_c$. A joint solution of $a=b=0$ is given by the ``super critical'' point $\tau=\tau_c=\tilde{\tau}_c$, which
satisfies
\begin{equation}\label{eq:tautildecrit}
	\tilde{\tau}_c=\sqrt{\frac{2}{15}}\sqrt{1+\frininf^{2}},
\end{equation}
and depends on the amount of photons $\Ninmod$ in the input due to \eqref{eq:limfreqin}. The range $\frininf\in(0,1]$ corresponding to $\Ninmod \in [0,\infty)$ implies the following bounds on $\tilde{\tau}_c$:
\begin{equation}
	\tilde{\tau}_L < \tilde{\tau}_c \le \tilde{\tau}_R,
\end{equation}
where 
\begin{equation}\label{eq:tildetauLR}
	\tilde{\tau}_L \equiv \sqrt{\frac{2}{15}}, \quad \tilde{\tau}_R \equiv \frac{2}{\sqrt{15}}.
\end{equation}
Note that
\begin{equation}
   \tilde{\tau}_L<\tau_L <  \tilde{\tau}_R<\tau_R, 
\end{equation}
and therefore the bound becomes $\tau_L \le \tilde{\tau_c} \le \tilde{\tau}_R$ (see Fig.~\ref{fig:ytaudivision}). Then, the left bound $\tau_L$ together with Eq.~\eqref{eq:tautildecrit} leads to the following bound
on $\Ninmod$:
\begin{equation}\label{eq:Nimax}
   \Ninmod  \le \Ninmodcrit = \frac{1}{2} \left[ \sqrt{\frac{3}{2} + \frac{5}{\sqrt{12}}} - 1\right] \approx 0.3578.
\end{equation}
Now remember that we are looking to the point where $a=b=0$ and therefore $\tau=\tau_c=\tilde{\tau_c} \le \tilde{\tau}_R$. This condition together with Eqs.~\eqref{eq:tauc} and \eqref{eq:tildetauLR} implies the
following bound on $\Menv$:
\begin{equation}\label{eq:Menvmax}
	\Menv \le \Menvcrit = \frac{1}{2} \left[ \left(\sqrt{3} -\frac{2}{\sqrt{5}} \right)^{-1} - 1\right] \approx 0.0969.
\end{equation}
In Fig.~\ref{fig:ytaudivision} this corresponds to the interval of the line $y=y_c$ enclosed between $\tau_L$ and $\tilde{\tau}_R$.

Since we know that for $y<y_c$ scenario \ref{itm:oneMax} is realized and therefore, there is no saddle point. Hence, if $y>y_c$ and $\tau=\tau_c$ then a saddle point exist. When $y$ decreases and approaches $y_c$ the saddle point approaches $\frenv=0$ and disappears for $y<y_c$.

We further observe that for $y>y_c$ and small input energies $\Ninmod$ in the neighborhood of the critical
parameters $\tau_c, y_c$ the polynomial obtained by cutting the Taylor series in
Eq.~\eqref{eq:expansionCapGchi} up to the third power in $\frenv$ can develop two extrema, i.e. a maximum
followed by a minimum. Since the expansion given in Eq.~\eqref{eq:expansionCapGchi} can approximate
arbitrary well $\CapGchi$ in the vicinity of $\frenv = 0$ we deduce that the curve of $\CapGchi$ has to
present the two extrema as well. We confirm this numerically as seen Fig.~\ref{fig:capsvswenv} (b).

Note that for the classical additive noise channel $\tau=1$ and $y = \Menv$. Therefore, only by varying the value of $\Menv$ one may realize scenario \ref{itm:oneMax} or scenario \ref{itm:noExt} (monotonous decrease).

The two extrema (maximum and minimum) fall together to one saddle point at some $\frenv>0$ precisely when
\begin{equation}\label{eq:capsaddle}
	\begin{split}
		y = \tilde{y} & = (1-\tilde{\tau})\left(\Menv+\frac{1}{2}\right),\\
		\Menv & \le \Menvcrit = \frac{1}{2} \left[ \left(\sqrt{3} -\frac{2}{\sqrt{5}} \right)^{-1} - 1\right] \approx 0.0969,\\
		\Ninmod & \le \Ninmodcrit = \frac{1}{2} \left[ \sqrt{\frac{3}{2} + \frac{5}{\sqrt{12}}} - 1\right] \approx 0.3578,
	\end{split}	
\end{equation}
where $\tilde{\tau}$ is found by solving numerically $\rm{d} \CapGchi / \rm{d} \frenv = \rm{d}^2 \CapGchi /\rm{d} \frenv^2 = 0$ for $\frenv \in [0,\frthr]$ (see also~\cite{PLM09}). The parameters $\Menvcrit$ and
$\Ninmodcrit$ are the solutions of the equations
\begin{equation}
	\begin{split}
	  	y_c & = \left(1-\tilde{\tau}_R\right)\left(\Menvcrit+\frac{1}{2}\right),\\
		\tau_L & = \tilde{\tau}(\Ninmod=\Ninmodcrit).
	\end{split}	
\end{equation}
The domain $\Menv \le \Menvcrit$ in terms of $y$ corresponds to\\ $y \le (1-\tau)(\Menvcrit+1/2)$, see Figs.~\ref{fig:ytaudivision} and \ref{fig:ytildenmaxeta}. The zones of channel parameters corresponding to the different number of extrema are summarized in Table~\ref{table:wenvbehavior} and plotted in Figs.~\ref{fig:ytaudivision} and \ref{fig:ytildenmaxeta}. For the lossy and additive noise channel this was previously observed
in~\cite{PLM09,SKC10}.

\begin{figure}
	\centering
	\includegraphics[trim=12 0 0 0 mm, clip=true]{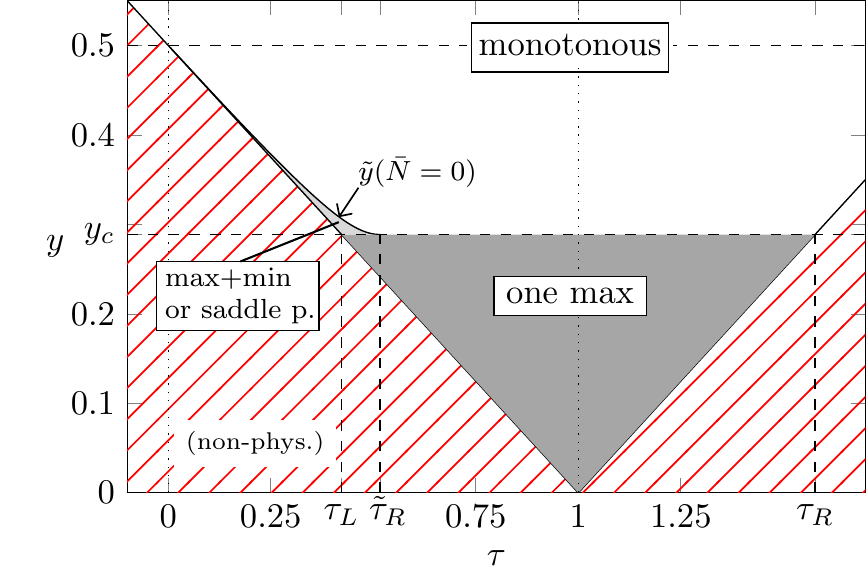}
	\caption{(Color online) Different zones, determined by $y$ and $\tau$, containing the scenarios \ref{itm:noExt} - \ref{itm:maxMin} of the Gaussian capacity $\CapGchi$ with respect to $\frenv$.
	- North-East lines ($y \le |1-\tau|/2$): zone where $y$ is non-physical.
	- Gray triangular region ($y < y_c$): the capacity exhibits one maximum.
	- Light gray region [$y_c < y < \tilde{y}(\Ninmod=0)$]: capacity may exhibit one minimum and one maximum, one saddle point or is monotonously decreasing with $\frenv$ (see Fig.~\ref{fig:ytildenmaxeta} for details).
	- White region ($y \ge \tilde{y}$, $y \ge y_c$): capacity is a monotonously decreasing function of $\frenv$.}
	\label{fig:ytaudivision}
\end{figure}

\begin{table}
	\center
	\begin{tabular}{|c|c||c|c|c|}\hline
	\multicolumn{2}{|c||}{\backslashbox[20mm]{$y \in$}{$\tau \in $}}  		& $(-\infty,0]$						&	$(0,\tilde{\tau}_R)$		&   $[\tilde{\tau}_R,\infty)$ 	\\ \hline\hline
	$[0,y_c)$  									&							& Monotonous						&	One max.					& 	One max.				\\ \hline
	\multirow{3}{*}{$[y_c,\infty)$}				& 	$[y_c,\tilde{y})$ 		& \multirow{3}{*}{Monotonous}		&	One max.+one min.		&	\multirow{3}{*}{Monotonous}							\\ \cline{2-2}\cline{4-4}
	  											&	$y=\tilde{y}$			& 									&   Saddle point				&								\\ \cline{2-2}\cline{4-4}
												&	$(\tilde{y},\infty)$	&									&	Monotonous				 	& \\ \hline
	\end{tabular}
	\caption[Extremality properties of $\CapGchi$ (single-mode) with respect to environment frequency $\frenv$]{Extremality properties of $\CapGchi$ with respect to $\frenv$.}
	\label{table:wenvbehavior}
\end{table}

Summarizing our results, we found the following behavior of the Gaussian capacity for the different channels.
\begin{itemize}
\item For the phase-conjugating channel $y \ge 1/2$ the linear term in the expansion is always
non-negative which means that the limiting value $\log_2(2\Ninmod+1)$ is always approached from below for
$\frenv \to 0$, and furthermore the second derivative of the Gaussian capacity is never zero because we
always have $\tau < 0 < \tau_L$. Hence, there is never a saddle point at $\frenv=0$. Indeed, our numerical
plots show that the capacity of the phase-conjugating channel is always monotonically decreasing
with $\frenv$ (see Fig.~\ref{fig:capsvswenv}).
\item The critical behavior $y = y_c$ can be observed for the amplification channel and the additive noise channel [$\tau \in [1,\infty)$]. However, the condition for a saddle point at $\frenv=0$ cannot be
fulfilled for those channels, since $\tau\ge 1> \tilde{\tau}_R$. In this case we observe numerically only two different behaviors: for $y < y_c$ a maximum appears at a finite $\frenv$ and for $y \ge y_c$ the capacity is always
monotonically increasing.
\item Only the lossy channel $(\tau \in [0,1])$ exhibits all four possible scenarios: in addition to one maximum for $y < y_c$ and a monotonous increasing for $y \ge \tilde{y}$, a saddle point appears when
Eqs.~\eqref{eq:capsaddle} are satisfied and one maximum followed by one minimum is present when inequality $y_c < y < \tilde{y}$ is satisfied.
\end{itemize}
The above results are depicted in Figs.~\ref{fig:ytaudivision} and \ref{fig:ytildenmaxeta}.

\section{Conclusions}\label{sec:conclusions}
We studied the Gaussian capacity of the fiducial Gaussian channel in a wide range of channel parameters including limiting cases. As shown previously our solution provides a way of calculating the Gaussian capacity
of any Gaussian channel for its given canonical decomposition~\cite{SKPPC13}. Two types of the solution are separated by an input energy threshold. If the input energy is above the threshold a thermal state at the
output can be realized, as well as a squeezed input state that exactly matches the squeezing of the environment state. This type of solution was known before. Below this threshold we find new solution given by a
transcendental equation. This type of solution was found first for the lossy and additive noise channels. Here we extended it to all types of Gaussian channels. We provides to a physical interpretation of the
solution choosing a particular representation for the covariance matrices so that the optimal modulation is found by jointly solving a Planck-like equation and a ``resonance'' equation. The ``resonance'' is
achieved exactly at the input energy threshold and holds above it. It provides the optimal input frequency matching the noise frequency, which implies that the modulated output state has a frequency equal to 1.

We investigated the capacity with respect to the noise squeezing for given transmissivity/gain and noise determinant and distinguished between four different behaviors for increasing noise squeezing (or decreasing
noise frequency): monotonically increasing, one maximum, a saddle point or one maximum followed by a minimum. We classified the precise boundaries for this behavior and concluded that only the lossy channel can
exhibit all four behaviors, whereas the capacity of the additive noise channel and amplifier can develop a maximum and the capacity of the phase-conjugating channel is always monotonically increasing with the noise
squeezing. 
\begin{figure}
	\centering
	\includegraphics[trim=15 0 0 0 mm, clip=true]{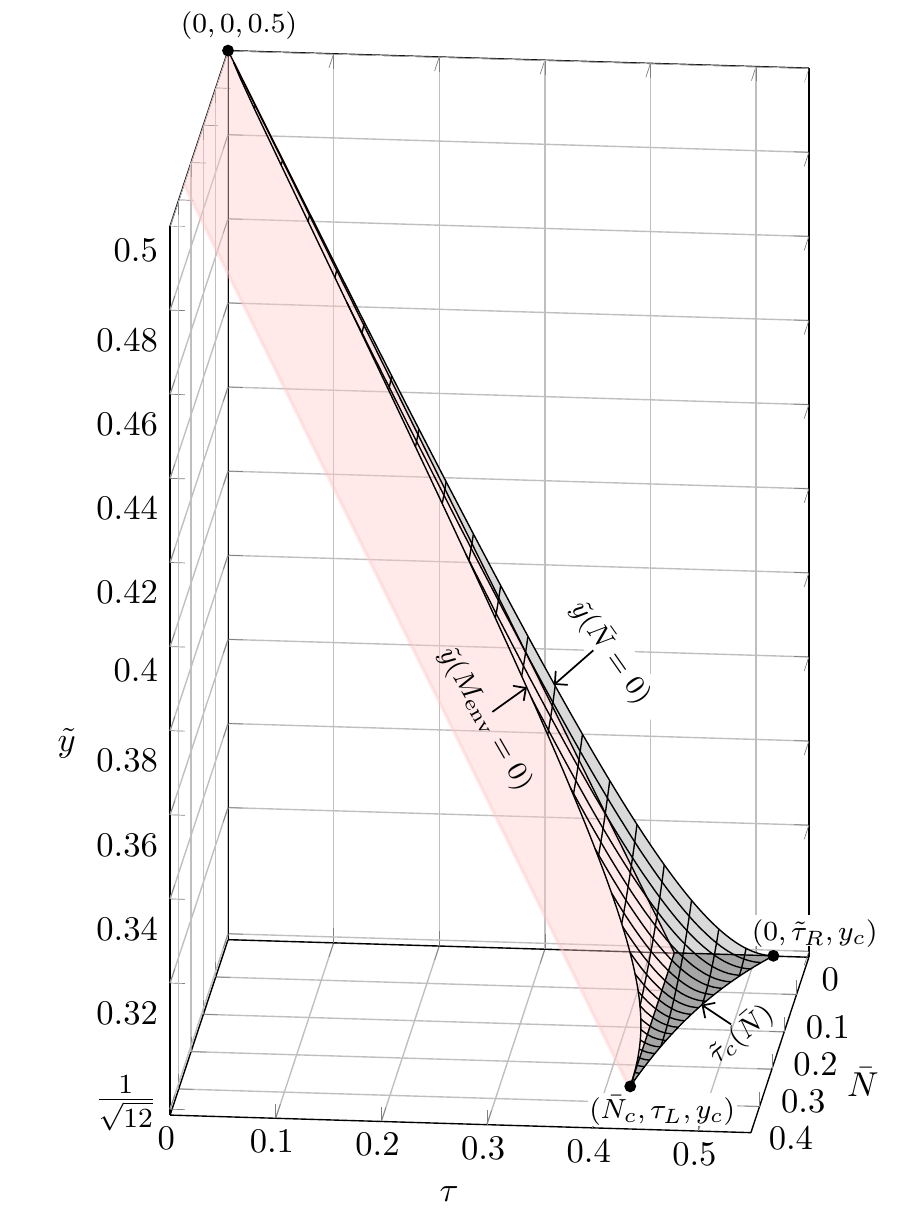}
	\caption{(Color online) Zones of different behavior for the fiducial Gaussian channel parametrized by $\tau$, $y$ and $\Ninmod$.} 
	\label{fig:ytildenmaxeta}
\end{figure}
Furthermore, we found that the Lagrange multiplier can be interpreted as the inverse temperature $\beta$ that enters in the Planck-like equation. From previous studies~\cite{PLM09, SDKC09, SKC11} where the solution
of a system of $n$ single-mode mode channels with joint energy constraint was discussed we know that this multiplier is common for all channels of such a system. This hints to the interpretation of a joint solution
for $n$ channels as corresponding to a kind of \emph{thermal equilibrium}, similar to the case of $n$ connected heat baths with different temperatures which thrive to a joint temperature $\beta^{-1}$.


%

\appendices
\section{Concatenated channels}\label{ap:concatenated}
In order to study the monotonicity of the classical capacity and Gaussian capacity on the parameter $\tau$ the following Lemma will be useful.
\begin{lemPipeline}\label{lemPipeline}
	For given parameters $\tau_1,\tau_2$ satisfying one of the three conditions:
	\begin{enumerate}
		\item $\tau_1,\tau_2 \ge 1$,
		\item $\tau_1,\tau_2 \in [0,1]$,
		\item $\tau_1 < 0, \tau_2 \in [0,1]$,
	\end{enumerate}
	the following equality holds:
	\begin{equation}\label{eq:chPipeline}
		\begin{split}
			& \chPS_{(\tau_2,|1-\tau_2|(\Menv+1/2),\frenv)} \circ \chPS_{(\tau_1,|1-\tau_1|(\Menv+1/2),\frenv)} \\
			&= \chPS_{(\tau_1\tau_2,|1-\tau_1\tau_2|(\Menv+1/2),\frenv)}.
		\end{split}		
	\end{equation}
\end{lemPipeline}
\begin{proof}
	In the following proof we use the parametrization 
	\begin{equation}
		\effPS = |1-\tau|\CMenv, \quad \CMenv = \left(\Menv+\frac{1}{2}\right)\begin{pmatrix}
			\frenv^{-1} & 0\\0 & \frenv
		\end{pmatrix}.
	\end{equation}
	The total output covariance matrix of the concatenated channel stated on the left hand side of Eq.~\eqref{eq:chPipeline} reads
	\begin{equation}\label{eq:CMoutLemma}
		\begin{split}
			\CMout 	& = |\tau_2|\left(|\tau_1| \CMin^{(\pm)} + |1-\tau_1| \CMenv\right) + |1-\tau_2| \CMenv,\\
					& = |\tau_1\tau_2| \CMin^{(\pm)} + \left( |\tau_2| |1-\tau_1| + |1-\tau_2|\right) \CMenv ,
		\end{split}
	\end{equation}
	where $\CMin^{(+)} = \CMin$ and $\CMin^{(-)}$ has an additional minus sign in front of its off-diagonal elements. It is straight forward to show that for all three conditions stated in the Lemma the second line of Eq.~\eqref{eq:CMoutLemma} can be simplified to
	\begin{equation}
		\CMout = |\tau_1\tau_2| \CMin^{(\pm)} + |1-\tau_1\tau_2| \CMenv,
	\end{equation}
	which completes the proof.
\end{proof}



\ifCLASSOPTIONcaptionsoff
  \newpage
\fi

\end{document}